%% file: fullResRspPaper.tex
\documentclass[prodmode]{acmsmall_TR150505}	
\input{UsefulLatexStuffFull}

\hypersetup{
  pdftitle = {A Characterization of the Complexity of Resilience and Responsibility for Self-join-free Conjunctive Queries},
  pdfkeywords = {deletion propagation, resilience, responsibility},
  pdfauthor = {Cibele Freire & Wolfgang Gatterbauer & Neil Immerman & Alexandra Meliou},
}

\begin{document}

\title{A Characterization of the Complexity of Resilience and Responsibility for Self-join-free Conjunctive Queries}

\author{
Cibele Freire
\affil{University of Massachusetts, Amherst}
Wolfgang Gatterbauer
\affil{Carnegie Mellon University}
Neil Immerman
\affil{University of Massachusetts, Amherst}
Alexandra Meliou
\affil{University of Massachusetts, Amherst}
}

\input{full_0_Abstract}
\maketitle
\input{full_1_Introduction}

\input{full_2_Background}

\input{full_3_Resilience}

\input{full_4_Fds}

\input{full_5_Responsibility}

\input{full_6_FelatedWork}

\input{full_7_Discussion}

\bibliography{bib/causation}
\newpage
\appendix
\input{full_Nomenclature}

\end{document}

%% file: UsefulLatexStuffFull.tex
\usepackage{mathtools}	
\usepackage{graphicx}   
\usepackage{verbatim}   
\usepackage[usenames,table,dvipsnames]{xcolor}      
\usepackage{multirow}

\usepackage{balance}
\usepackage{tikz}
\usetikzlibrary{arrows,shapes,trees,backgrounds,automata}
\usepackage{tkz-euclide}
\usetkzobj{all}
\usepackage{tikzsymbols}

\usepackage{caption}

\usepackage{subfig}
\usepackage{bm}                     	

\usepackage{enumitem}   
\setlist{noitemsep}
\setlist{nolistsep}

\usepackage{booktabs}       
\usepackage{tabularx}

\usepackage[algoruled, lined]{algorithm2e}
\usepackage{aliascnt}  		

\usepackage{mathdots}
\usepackage{hyperref}   
\hypersetup{                                                            
	pdfpagemode=UseNone,               
    pdfstartview=Fit,
    pdfpagelayout=SinglePage,       
    bookmarks=false,
    bookmarksnumbered=true,
    bookmarksopen=false,             
    pdftex=true,
    unicode,
	colorlinks=false,
	pdfstartview={FitH},
	pageanchor=false
}

\definecolor{dg}{cmyk}{0.60,0,0.88,0.27}
\usepackage{hhline}

\newaliascnt{conjecture}{theorem}			
\newtheorem{conjecture}[conjecture]{Conjecture}    
\aliascntresetthe{conjecture}  				

\newaliascnt{fact}{theorem}				
\newtheorem{fact}[fact]{Fact}  				
\aliascntresetthe{fact}	  				

\newaliascnt{observation}{theorem}				
\newtheorem{observation}[observation]{Observation}  				
\aliascntresetthe{observation}

\newaliascnt{myalgorithm}{theorem}				
\newtheorem{myalgorithm}[myalgorithm]{Algorithm}  				
\aliascntresetthe{myalgorithm}

\numberwithin{equation}{section}
\makeatletter
\let\c@equation\c@theorem
\makeatother

\providecommand{\exampleautorefname}{Example}

\renewcommand{\sectionautorefname}{Section}
\renewcommand{\subsectionautorefname}{Subsection}

\let\orgautoref\autoref   	
\providecommand{\Autoref}[1]{%
	\def\figureautorefname{Figure}%
	\def\subfigureautorefname{Figure}%
    	\def\exampleautorefname{Example}%
    	\def\sectionautorefname{Section}%
    	\def\subsectionautorefname{Section}%
    	\def\algorithmautorefname{Algorithm}%
    	\def\resultautorefname{Result}%
	\def\factautorefname{Fact}%
    	\orgautoref{#1}%
}
\renewcommand{\autoref}[1]{
    \def\figureautorefname{Fig.}
    \def\subfigureautorefname{Fig.}
    \def\sectionautorefname{Section}
    \def\subsectionautorefname{Section}
    \def\subsubsectionautorefname{Section}
    \def\exampleautorefname{Example}
    \def\resultautorefname{Result}
    \def\factautorefname{Fact}
    \orgautoref{#1}
}
\renewcommand{\autoref}[1]{%
    \def\figureautorefname{Fig.}%
    \def\subfigureautorefname{Fig.}%
    \def\sectionautorefname{Section}%
    \def\subsectionautorefname{Section}%
    \def\subsubsectionautorefname{Section}%
    \def\exampleautorefname{Example}%
    \def\resultautorefname{Result}%
    \def\factautorefname{Fact}%
    \orgautoref{#1}%
}

\newcommand{\specificref}[2]{\hyperref[#2]{#1~\ref*{#2}}}

\newcommand{\sat}{\mbox{{\rm\sc SAT}}}

\newcommand{\set}[1]{\{#1\}}

\renewcommand{\vec}[1]{\boldsymbol{\mathbf{#1}}}

\newcommand{\datarule}{{\,:\!\!-\,}}			

\newcommand{\true}{\texttt{true}\xspace}		
\newcommand{\false}{\texttt{false}\xspace}

\newcommand{\smallsection}[1]{\vspace{5mm}\noindent\textbf{#1.}}	
\newcommand{\introparagraph}[1]{\textbf{#1.}}

\newcommand{\rewrite}{\leadsto}
\newcommand{\rewriteStar}{\stackrel{\star}{\leadsto}}

\newcommand{\enSymb}{\textup{n}}
\newcommand{\exSymb}{\textup{x}}
\newcommand{\en}[1]{{#1}^{\enSymb}}	
\newcommand{\ex}[1]{{#1}^{\exSymb}}

\usepackage{bigstrut}				

\newcommand{\var}{\textup{\texttt{var}}}      	

\ifpdf
\DeclareGraphicsExtensions{.pdf, .jpg, .tif}
\else
\DeclareGraphicsExtensions{.eps, .jpg}
\fi

\newcommand{\rsp}{\mathtt{RSP}} 
\newcommand{\res}{\mathtt{RES}} 
\newcommand{\rats}{q_{\textrm{rats}}}  
\newcommand{\raxx}{q_{\textrm{r$^x$at$^x$s}}}  
\newcommand{\brats}{q_{\textrm{brats}}}  
\newcommand{\brx}{q_{\textrm{br$^x$ats}}}
\newcommand{\brxxx}{q_{\textrm{br$^x$at$^x$s$^x$}}}
 
\newcommand{\ptime}{$\mathsf{PTIME}$} 

\newcommand{\NP}{\textup{\textsf{NP}}\xspace}
 
\newcommand{\PTIME}{\textup{\textsf{PTIME}}\xspace} 
\newcommand{\source}{\mathtt{DP_{source}}} 
\newcommand{\view}{\mathtt{DP_{view}}}

\newcommand{\Tri}{\textup{\textrm{T}}}

\renewcommand{\vec}[1]{\boldsymbol{\mathbf{#1}}}

\renewcommand{\phi}{\varphi}      
\renewcommand{\angle}[1]{ \langle #1 \rangle }
\newcommand{\ov}{\overline}

\usepackage{wasysym}

\newcommand{\fd}{\varphi}
\newcommand{\setfd}{\Phi}

\newcommand{\bigset}[2]{\bigl\{ #1 \,\bigm|\, #2 \bigr\} }

\newcommand{\qLra}{\quad\Leftrightarrow\quad}
\newcommand{\sland}{\;\land\;}
\newcommand{\abs}[1]{ \vert #1 \vert }

\newcommand{\dom}{\textrm{dom}}

\renewenvironment{thebibliography}[1]{
    \begin{oldthebibliography}{#1}
      \small
        \setlength{\itemsep}{1.5pt}
}
{
\end{oldthebibliography}
}

\newcounter{myeqn}
\newcommand{\myequation}[2]{\par\refstepcounter{theorem}
\setcounter{myeqn}{\value{theorem}}\vspace{4mm}\label{#1}
\noindent\makebox[\textwidth]{\hfill${\displaystyle #2 }$\hfill
{\mbox{\rm (\ref{#1})}}}
\par\smallskip\vspace{2mm}}

\graphicspath{{figs/}}

%% file: full_0_Abstract.tex
\begin{abstract}	

Several research thrusts in the area of data management have focused on
understanding how changes in the data affect the output of a view or standing
query. Example applications are explaining query results, propagating updates
through views, and anonymizing datasets. These applications usually rely on
understanding how interventions in a database impact the output of a
query.
An important aspect of this analysis is the problem of \emph{deleting a
minimum number of tuples from the input tables} to make a given Boolean query
false. 
We refer to this problem as ``\emph{the resilience of a query}'' and show its  connections to the well-studied problems of deletion propagation and causal responsibility.
In this paper, we study the complexity of resilience for self-join-free
conjunctive queries, and also make several contributions to previous known results
for the problems of \emph{deletion propagation with source side-effects} and 
\emph{causal responsibility}:
(1)~We define the notion of resilience and provide a complete dichotomy for the class of
\emph{self-join-free conjunctive queries with arbitrary functional
dependencies}; this dichotomy also extends and generalizes previous
tractability results on deletion propagation with source side-effects.
(2)~We formalize the connection between resilience and causal responsibility,
and show that resilience has a larger class of tractable queries than
responsibility. (3)~We identify a mistake in a 
previous dichotomy for the
problem of causal responsibility and offer a revised characterization based on new, simpler, and more intuitive notions.
(4)~Finally, we
extend the dichotomy for causal responsibility in two ways: (a)~we treat cases where the input tables
contain functional dependencies, and (b)~we compute responsibility for a set of tuples specified via wildcards.
\end{abstract}

%% file: full_1_Introduction.tex
\section{Introduction}\label{sec:intro}

As data continues to grow in volume, the results of relational queries become
harder to understand, interpret, and debug through manual inspection. Data
management research has recognized this fundamental need to derive
\emph{explanations for query results} and 
explanations for 
surprising observations.
Existing work has defined explanations as predicates in a
query~\cite{DBLP:journals/pvldb/0002M13,SudeepaSuciu14,DBLP:conf/sigmod/ChapmanJ09}, or as modifications to the input
data~\cite{MeliouGMS11,DBLP:journals/pvldb/HuangCDN08,DBLP:journals/pvldb/HerschelHT09}. 
In the
latter category, the metric of \emph{causal responsibility}, first introduced by \citeN{ChocklerH04}, quantifies the contribution of an input tuple to a
particular output. One can then derive explanations by 
\emph{ranking input tuples} using their responsibilities: tuples with high degree of
responsibility are better explanations for a particular query result than tuples with low responsibility~\cite{MeliouGMS11}.

\begin{figure}
\setlength{\arrayrulewidth}{0.6pt}

\newcommand{\cellLeft}[1]{\multicolumn{1}{|>{\columncolor{black!10}[.95\tabcolsep][.95\tabcolsep]}l|}{#1}}
\newcommand{\cellRight}[1]{\multicolumn{1}{>{\columncolor{black!10}[.95\tabcolsep][.95\tabcolsep]}l|}{#1}}

\centering
\subfloat[\mbox{\small{Source-side effects: $ \min |\Gamma|$}}]{\makebox[1.1\width][c]{\includegraphics[scale=0.38]{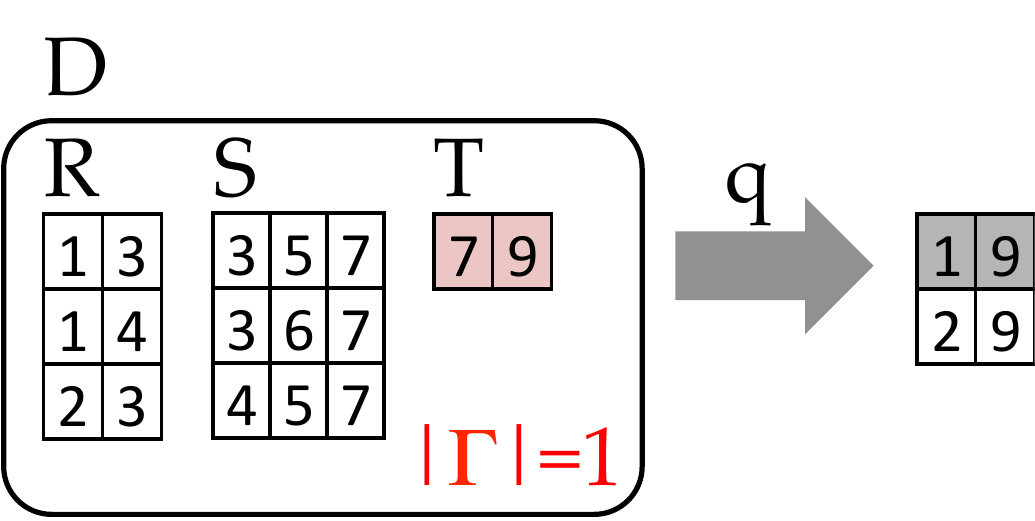}
\label{Fig_SourceSideEffect}}}
\hspace{2cm}
\subfloat[\small{Resilience: $\min |\Gamma|$}]{\makebox[1.1\width][c]{\includegraphics[scale=0.38]{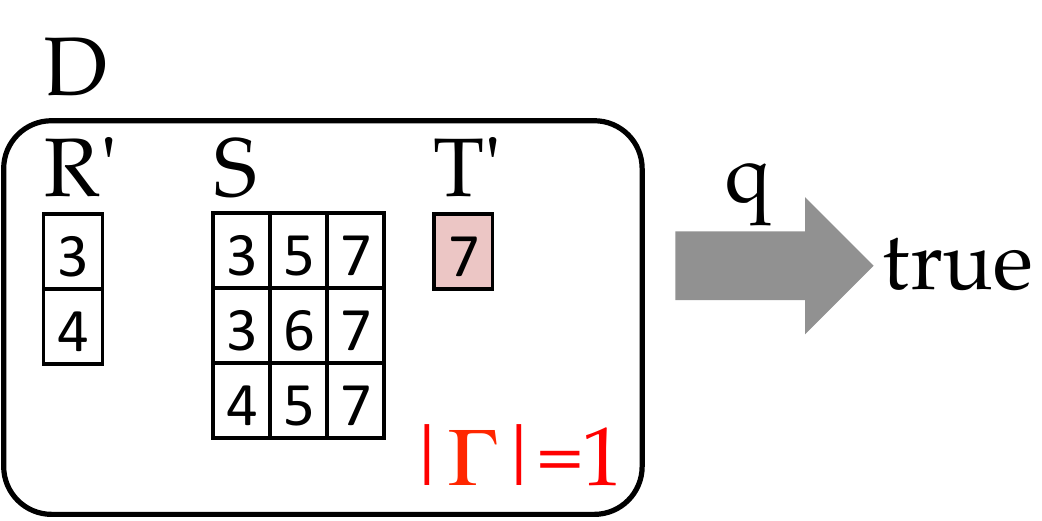}
\label{Fig_Resilience}}}

\subfloat[\small{Responsibility: $\min |\Gamma|$}]{\makebox[1.1\width][c]{\includegraphics[scale=0.38]{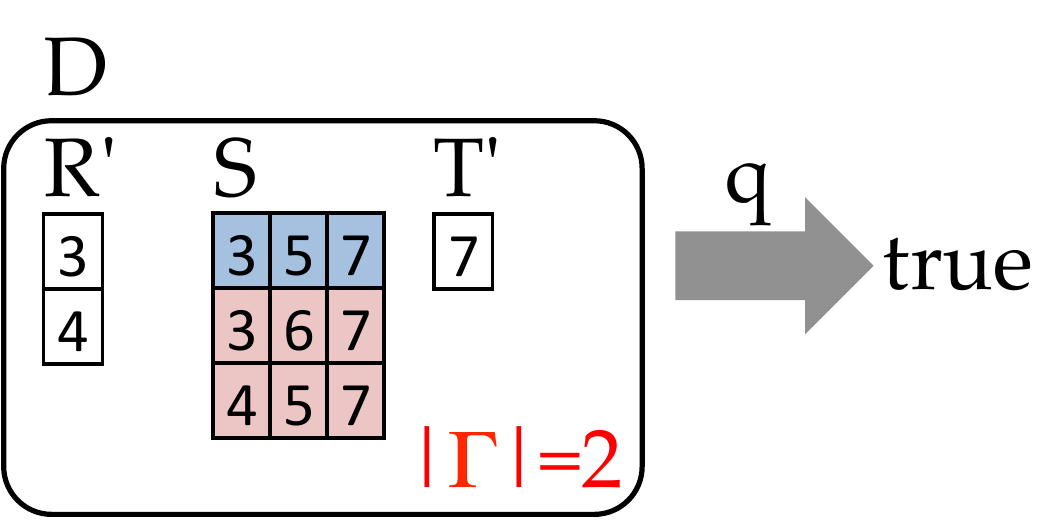}
\label{Fig_CausalityResponsibility}}}
\hspace{2cm}
\subfloat[\small{View-side effects: $\min |\Delta|$}]{\makebox[1.1\width][c]{\includegraphics[scale=0.38]{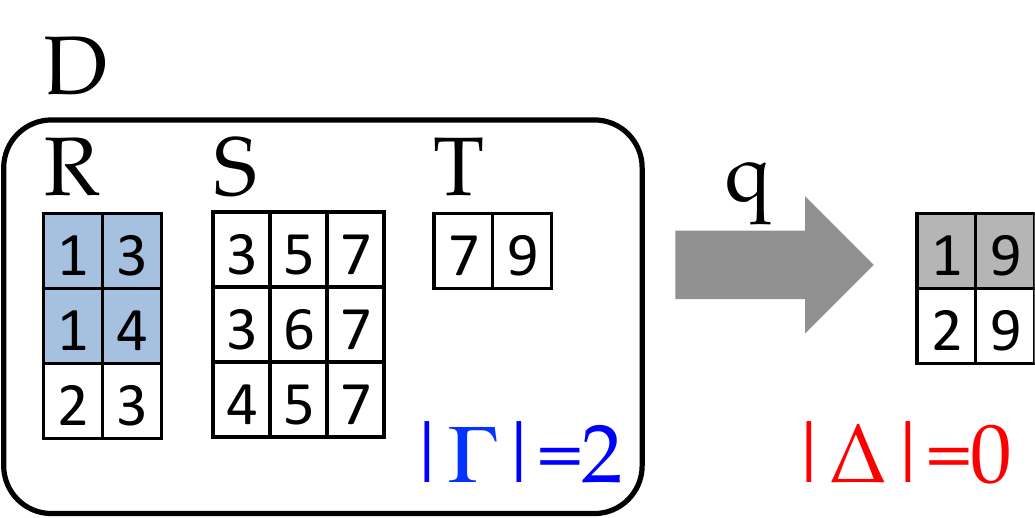}
\label{Fig_ViewSideEffect}}}

\subfloat[\small{Source-side effect problem: prior and our dichotomy results}]{
\setlength{\tabcolsep}{1.1mm}
\renewcommand{\arraystretch}{1.2}
\small
\centering
\vspace{3mm}
\begin{tabular}{ | m{62mm} | c | c | c | c | c|}
    \hline
    \emph{SJ}: Queries with selections and joins
        & \ptime 
        & \multirow{2}{*}{\cite{Buneman:2002}} \\
    \cline{1-2}
    \emph{PJ}: Queries with projections and joins
        & \NP-complete      
        &           \\  
    \hline
    \hline
    \emph{``Key-preserving'' SPJ} queries                       
        & \ptime    
        & \multirow{2}{*}{\cite{Cong12}} \\
    \cline{1-2}
    All other SPJ queries   
        & \NP-complete  
        &       \\
    \hline
    \hline
    \cellLeft{\emph{``Triad-free''} SPJ queries}
        & \cellRight{\ptime}     
        & \cellRight{}  \\
	\hhline{|-|-|>{\arrayrulecolor{black!10}}->{\arrayrulecolor{black}}|}	
	\cellLeft{All other SPJ queries}
        & \cellRight{\NP-complete}    
        & \cellRight{\multirow{-2}{*}{this paper}}       \\
    \hline
    \hline
    \cellLeft{\emph{``FD-induced triad-free''} SPJ queries}
        & \cellRight{\ptime}     
        & \cellRight{}  \\
	\hhline{|-|-|>{\arrayrulecolor{black!10}}->{\arrayrulecolor{black}}|}
    \cellLeft{All other SPJ queries}
        & \cellRight{\NP-complete}
        & \cellRight{\multirow{-2}{*}{this paper}}       \\
    \hline  
\end{tabular}
}

\subfloat[\small{View-side effect problem: prior dichotomy results}]{
\setlength{\tabcolsep}{1.1mm}
\renewcommand{\arraystretch}{1.2}
\small
\centering
\vspace{3mm}
\begin{tabular}{ | m{62mm} | c | c | c | c | c|}
    \hline
    \emph{SJ}: Queries with selections and joins
        & \ptime 
        & \multirow{2}{*}{\cite{Buneman:2002}} \\
    \cline{1-2}
    \emph{PJ}: Queries with projections and joins
        & \NP-complete      
        &           \\  
    \hline
    \hline
    \emph{``Key-preserving'' SPJ} queries                       
        & \ptime    
        & \multirow{2}{*}{\cite{Cong12}} \\
    \cline{1-2}
	All other SPJ queries   
        & \NP-complete  
        &       \\
    \hline
    \hline
    \emph{``Head-dominated'' SPJ} queries                       
        & \ptime    
        & \multirow{2}{*}{\cite{KimelfeldVW12}} \\
    \cline{1-2}
    All other SPJ queries   
        & \NP-complete  
        &       \\
    \hline  
    \hline
    \emph{``Functional head-dominated'' SPJ} queries
        & \ptime    
        & \multirow{2}{*}{\cite{Kimelfeld12}} \\
    \cline{1-2}
    All other SPJ queries   
        & \NP-complete  
        &       \\
    \hline  
\end{tabular}
}

\caption{This paper contains dichotomy results for (a) deletion propagation
with source-side effects, (b) resilience, and (c) responsibility for
causality. Besides others, they imply a complete dichotomy for the source
side-effect problem for the class of \emph{self-join-free conjunctive queries
in the presence of functional dependencies} (e). Thus, this part of our work
is similar in scope to \protect\cite{KimelfeldVW12} and
\protect\cite{Kimelfeld12} for the problem of view-side effects (f). We derive
these results by analyzing a simpler concept: the resilience of Boolean
queries. In addition (not shown in the figure), we provide a correction to a
prior dichotomy result for causal responsibility and then extend it in two
ways: responsibility for tables with functional dependencies and
responsibility for tuples with wildcards, e.g., $S(*,5,7)$. 
}
\label{Fig_SourceViewSideEffectProblems}
\end{figure}

A seemingly unrelated notion, the concept of \emph{deletion propagation with source side-effects}~\cite{Buneman:2002}, seeks a minimum set of tuples in the input tables that should be deleted from the database in order to delete a particular tuple from a query. 
Query results that have a larger set of tuples that need to be deleted are more reliable or more ``robust''  to changes in the input database than others.
This measure of relative importance can provide another type of explanation and allows us to  \emph{rank the output tuples} by their relative robustness.

In this paper, we take a step back and re-examine how particular
interventions (\emph{tuple deletions in the input of a query}) impact
its output. Specifically, we study how ``resilient'' a Boolean query is with
respect to such interventions. \emph{Resilience} identifies the smallest number of
tuples to delete from the input to make the query false. We will show that characterizing the complexity of this problem also allows us to study the complexities of 
\emph{both} deletion propagation with source side-effects and causal responsibility with minor modifications.

\smallsection{Deletion propagation and existing results}
Databases allow users to interact with data through views, which are often
conjunctive queries. Views can be used to simplify complex queries, enforce
access control policies, and preserve data independence for external
applications.
Of particular interest is how deletions in the input data affect the view
(which is a trivial problem), but also how deletions in the view could be
achieved by appropriately chosen deletions in the input data (which is far
less trivial).
Concretely, the problem of \emph{deletion propagation}~\cite{Buneman:2002,Dayal82} seeks a set $\Gamma$ of tuples in the input tables that should be
deleted from the database in order to delete a particular tuple from the view.
Intuitively, this deletion should be achieved with \emph{minimal
side-effects}, 
where side-effects are defined with either of two objectives: 
(a) deletion propagation with \emph{source
side-effects} ($\source$) seeks a minimum set of input tuples $\Gamma$ in order to delete a given
output tuple; whereas (b) deletion propagation with \emph{view side-effects} ($\view$)
seeks a set of input tuples $\Gamma$ that results in a minimum number of output tuple
deletions in the view, other than the tuple of interest~\cite{Buneman:2002}.

\begin{example}[Source \& View side effects]\label{ex:introSourceView}
Consider the query 
\[q(x,u) \datarule R(x,y), S(y,z,w), T(w,u)\]
defining a
view over the database $R,S,T$ shown below. To delete tuple $v_1$
from the resulting view with minimum source side-effects, one only needs to
remove tuple $t_1$ from the database. Therefore, the optimal solution to
$\source$ is $\Gamma = \{t_1\}$ with $|\Gamma|=1$ (see~\autoref{Fig_SourceSideEffect}).

However, the deletion of
$t_1$ also removes $v_2$, which is a view side-effect: $\Delta=\{v_2\}$
with $|\Delta|=1$. The optimal solution to $\view$, which minimizes the
side-effects on the view (set $\Delta$) is the set of input tuples
$\Gamma=\{r_1, r_2\}$: deleting these two tuples removes only $v_1$ from the
view but not $v_2$, and thus has no view-side effects, i.e., $\Delta =
\emptyset$ with $|\Delta| = 0$ (see~\autoref{Fig_ViewSideEffect}).

\begin{center}
    \small
    \renewcommand\tabcolsep{3pt}
    \renewcommand{\arraystretch}{1.2}
    	\centering
    	\begin{tabular}{r|c|c|  c  r|c|c|c|  c  r|c|c|  p{5mm}  r|c|c|	c c}
    		\multicolumn{1}{c}{}
    		&\multicolumn{2}{l}{$R$} & \multicolumn{2}{l}{\hspace{5mm} } 
    		&\multicolumn{3}{l}{$S$} & \multicolumn{2}{l}{\hspace{5mm} } 
    		&\multicolumn{2}{l}{$T$} & \multicolumn{2}{l}{\hspace{12mm} } 
    		&\multicolumn{2}{l}{$q$} \\
    		\cline{2-3}\cline{6-8}\cline{11-12}\cline{15-16}
    		&$X$ &$Y$ 
    			&\multicolumn{2}{c|}{} &$Y$ &$Z$ &$W$ 
    			&\multicolumn{2}{c|}{} &$W$ &$U$ 
    			&\multicolumn{2}{c|}{} &$X$ &$U$\\
    		\cline{2-3}\cline{6-8}\cline{11-12}\cline{15-16}
    		$r_1$	&1	&3   &&$s_1$	&3	&5	&7	&&$t_1$	&7	&9	&&$v_1$&	1	&9\\
    		\cline{11-12}
    		$r_2$	&1	&4   &&$s_2$	&3	&6	&7	&\multicolumn{5}{c}{ }		&$v_2$	&2	&9\\
    		\cline{15-16}	
    		$r_3$	&2	&3   &&$s_3$	&4	&5	&7	\\	
    		\cline{2-3}\cline{6-8}	
    	\end{tabular}
\end{center}

\end{example}

\emph{Known complexity results.}
\citeN{Buneman:2002} showed that both
variants are in general \NP-complete for conjunctive queries containing
projections and joins (PJ), whereas they are in \ptime\ for queries containing
only selections and joins (SJ). Later, \citeN{Cong12} identified a
class of PJ queries, called ``key-preserving,'' for which both problem
variants can be solved in \ptime. According to these two results, the
query from \specificref{Example}{ex:introSourceView} falls into the general class of
\NP-complete queries.

In addition, \citeN{KimelfeldVW12} provided a more refined
dichotomy result for the problem of minimal view side-effects for
self-join-free conjunctive queries (CQs). This dichotomy leads to more
polynomial time cases, as it characterizes the complexity based on a property
of the query structure (using the property of ``\emph{head domination}''),
rather than high-level database operators (e.g., projections and joins). For
example, the query of \specificref{Example}{ex:introSourceView} is not head-dominated,
which means that $\view$ is indeed \NP-complete for that query. Later work has
also extended the dichotomy result to self-join-free CQs with functional
dependencies (FDs)~\cite{Kimelfeld12}.

\smallsection{Causal responsibility and existing results}
The problem of causal
responsibility~\cite{MeliouGMS11} seeks, for a
given query \emph{and a specified input tuple}, a minimum set of other input
tuples $\Gamma$ that, if deleted would make the tuple of interest
``counterfactual,'' i.e., the query would be true with that tuple present,
or false if the tuple was also deleted.
Both problems of resilience and of causal responsibility rely on the notion of minimal interventions in the input database and are thus closely related.
However, we will show that resilience is easier (has lower
complexity) than responsibility, and provide extensive discussion of the
connections among all these related problems.

\begin{example}[Resilience \& Causal responsibility]\label{ex:introResilience}
Consider again the query from \specificref{Example}{ex:introSourceView} and the output
tuple $v_1=(1,9)$. Applying the substitution $[(x,u) / (1,9)]$, i.e., substituting the variables $x$ and $u$ with 1 and 9,
respectively, we get a query $q(1,9) \datarule R(1,y),$ $S(y,z,w), T(w,9)$.
The solution to $\source$ for $q$ and tuple $v_1$ is then equivalent to the
solution of the resilience problem over the Boolean query $q' \datarule R'(y),
S(y,z,w), T'(w)$ over the database $R',S,T'$ with 
$R'(y) \datarule R(1,y)$ and 
$T'(w) \datarule T(w,9)$ shown below. The answer to
the resilience problem for $q'$ is $\Gamma = \{t_1'\}$ with $|\Gamma|=1$:
deleting tuple $t_1'$ makes the query false (also see~\autoref{Fig_Resilience}).
\begin{center}
    \small
    \renewcommand\tabcolsep{3pt}
    \renewcommand{\arraystretch}{1.2}
    	\begin{tabular}{r|c|  c  r|c|c|c|  c  r|c|  p{5mm}  r|c|c|	c c}
    		\multicolumn{1}{c}{}
    		&\multicolumn{1}{l}{$R'$} & \multicolumn{2}{l}{\hspace{5mm} } 
    		&\multicolumn{3}{l}{$S$} & \multicolumn{2}{l}{\hspace{5mm} } 
    		&\multicolumn{1}{l}{$T'$} & \multicolumn{2}{l}{\hspace{20mm} } \\
    		\cline{2-2}\cline{5-7}\cline{10-10}
    		&$Y$ 
    			&\multicolumn{2}{c|}{} &$Y$ &$Z$ &$W$ 
    			&\multicolumn{2}{c|}{} &$W$ \\ 
    		\cline{2-2}\cline{5-7}\cline{10-10}
    		$r_1'$	&3   &&$s_1$	&3	&5	&7	&&$t_1'$	&7	\\
    		\cline{10-10}
    		$r_2'$	&4   &&$s_2$	&3	&6	&7	\\
    		\cline{2-2}	
    			\multicolumn{3}{c}{}  &$s_3$	&4	&5	&7	\\	
    		\cline{5-7}	
    	\end{tabular}
\end{center}

The causal responsibility problem requires a tuple in the lineage of the query
as additional input. For example, the responsibility of tuple $s_1$ in query
$q'$ corresponds to the contingency set $\Gamma = \{s_2,s_3\}$ with $|\Gamma|
= 2$. Deleting these two tuples makes $s_1$ a counterfactual cause for $q'$,
i.e., the query is true if $s_1$ is present or false, otherwise (also see~\autoref{Fig_CausalityResponsibility}).
\end{example}

\emph{Known complexity results.}
\citeN{MeliouGMS11} showed that causality of a given tuple can be
computed in polynomial time for any conjunctive query. Further, that work
presented a dichotomy result for computing causal responsibility for
self-join-free conjunctive queries, based on a characterization of a query
property called \emph{weak linearity}. 
However, in this work, we identify an error in the existing dichotomy which
classified certain hard queries into the polynomial class of queries. In
particular, we found that the existing notion of ``domination'' is not
sufficient to characterize the dichotomy and we provide here a refinement of
domination called ``full domination'' that together with a new concept of
``triads'' solves this issue.

\smallsection{Contributions of our work}
In this paper, we study the problem of minimal interventions with respect to a new notion called
\emph{resilience} of a Boolean query, which is a minimum number of input
tuples that need to be deleted in order to make the query false. A method that
provides a solution to resilience can immediately also provide an answer to
the \emph{deletion propagation with source-side effects problem} by defining a new Boolean query and database, replacing all head variables in the view with constants of the
output tuple. We define our results in terms of ``resilience'' since the
notion of resilience has obvious analogies to universally known minimal set
cover problems.  
At the same time, our complexity results on resilience also allow us to study the problem of \emph{causal responsibility}.
We thus state our contributions
with respect to both deletion propagation and causal responsibility.

\smallskip
\introparagraph{(1) Contributions to deletion propagation}
Our results on resilience imply a refinement for the complexity of
\emph{minimum source side-effects} by defining a novel, yet simple and
intuitive property of the query structure called ``\emph{triads}.''
For the class of self-join-free conjunctive queries, we show that resilience
is \NP-complete if the query contains this structure, and \ptime\ otherwise (\Autoref{sec:resilience}).
Determining whether a query contains a triad can be done very efficiently, in
polynomial time with respect to query complexity.
This implies that $\source$ can always be
solved in \ptime\ for the query of \specificref{Example}{ex:introSourceView}. These results
are analogous to the results of \citeN{KimelfeldVW12} for the view-side effect
problem. In addition, our dichotomy criterion also allows the specification of
``forbidden'' tables (called \emph{exogenous} tables) that do not allow
deletions. This is an extension to the traditional definition of the deletion
propagation problem and affects the complexity of queries in non-obvious ways
(defining a table as exogenous can make both easy queries hard, and hard
queries easy).

Our work also provides a \emph{complete dichotomy result for the class of
self-join-free CQs with Functional Dependencies} (\Autoref{sec:FDs}). These
results are analogous to the results of \citeN{Kimelfeld12} for the view-side
effect problem. At a high-level, we define rewrite steps that are induced by
the functional dependencies, and check the resulting query for the presence of
triads.

In particular, our dichotomy result on the resilience of a Boolean
conjunctive query provides \emph{new tractable solutions to the otherwise hard
minimum hypergraph vertex cover problem}. Our \ptime\ classes for resilience
define families of hypergraphs for which minimum vertex cover is also always
in \ptime. As such, resilience provides an intuitive definition that can
draw analogies to problems even outside the database community.  However, these implications are outside the scope of this paper.

\smallskip
\introparagraph{(2) Contributions to causal responsibility}
We show that responsibility is a more fine-grained notion than resilience,
resulting in higher complexity. In particular, we show query $\rats$ in \autoref{fig:ratsHypergraph} for which resilience is in \ptime\ (\specificref{Cor.}{res(rats) easy}), whereas responsibility is \NP-complete (\specificref{Prop.}{responsibility of rats is hard}).
The benefit of responsibility is
that it allows us to \emph{rank input tuples} based on their impact to a
query, thus making it applicable to settings where this ranking is
important, such as providing explanations and data compression (by compressing
data with small contributions to an output). In \Autoref{sec:outlook}, we
discuss ways to use resilience in these applications, and thus benefit from
its reduced complexity compared to responsibility.

In addition, we found that responsibility is a more subtle concept than we
previously thought. In particular, we identified an error in the existing
dichotomy for responsibility~\cite{MeliouGMS11} which
classified certain hard queries into the polynomial class of queries.
In particular, we found that the existing notion of ``domination'' is not sufficient to characterize the dichotomy. In
\Autoref{responsibility sec}, we provide a refinement of domination called ``full domination'' that helps use solve this issue. 
In addition, our new results provide two significant extensions to the
previous dichotomy: 
(a)~We generalize the notion of responsibility from simple tuples to tuples with wildcards.
(b)~We show that through a process of query
rewrites, our dichotomy results continue to hold in the presence of functional dependencies
over the input relations.

\smallskip
\introparagraph{Outline}  
\Autoref{sec:background} defines all notions mentioned here more formally and
discusses the connections of resilience with deletion propagation and causal
responsibility. \specificref{Sections}{sec:resilience} and~\ref{sec:FDs} contain our
two main technical contributions for the problem of resilience, while
\Autoref{responsibility sec} corrects the dichotomy of responsibility and
extends it to the case of tuples with wildcards and functional dependencies.
\Autoref{sec:relatedWork} reviews additional related work, and
\Autoref{sec:outlook} discusses implications, open problems, and future
directions.

%% file: full_2_Background.tex
\section{Formal setup and connections}\label{sec:background}
This section introduces our notation, defines
resilience, and formalizes the connections between the problems of resilience, deletion propagation, and causal responsibility. 

\smallskip
\introparagraph{General notations}
We use boldface (e.g., $\vec x = (x_1, \ldots, x_k)$) to denote
tuples or ordered sets.
A \emph{self-join-free conjunctive query} (sj-free CQ) is a first-order formula $q(\vec y) = \exists
\vec x\,(A_1 \wedge \ldots \wedge A_m)$ where 
the variables $\vec x = (x_1, \ldots, x_k)$ 
are called {\em existential variables},
$\vec y = (y_1, \ldots, y_c)$ are called the \emph{head variables} (or free variables),
and each atom $A_i$ represents a relation $R_i(\vec z_i)$ where $\vec z_i \subseteq \vec x \cup \vec y$.\footnote{We assume w.l.o.g.\ that
  $\vec z_i$ is a tuple of only variables without constants. This is so, because for any constant in
  the query, we can first apply a selection on each table and then consider the modified query with
  a column removed (see the transformation from resilience to source side-effects for
  details). }

The term ``self-join-free'' means that 
no relation symbol occurs more than once.
We write $\var(A_j)$ for the set of variables occurring in atom $A_j$.
The database instance is then the union of all tuples in the relations $D = \bigcup_i R_i$.
As usual, we abbreviate the query in Datalog notation by 
$q(\vec y) \datarule A_1, \ldots, A_m$.
For tuple $\vec t $, we write $D \models q[\vec t/\vec y]$ to denote
that $\vec t$ is in the query result of the non-Boolean query $q(\vec y)$ over database $D$.
The set of query results over database $D$ is denoted by $q(\vec y)^D$.

Unless otherwise stated, a \emph{query} in this paper denotes
a sj-free \emph{Boolean} conjunctive query  $q$ (i.e., $\vec y = \emptyset$). 
Because we only have sj-free CQ we do not have two atoms referring to the same relation,
so we may refer to atoms and relations interchangeably.
We write $D \models q$ to denote that the query $q$ evaluates to \true over the database
instance $D$, and $D \not\models q$ to denote that $q$ evaluates to \false. 
We call a valuation of all existential variables that is permitted by $D$ and that makes $q$ \true, a \emph{witness} $\vec w$.\footnote{Notice that our notion of witness slightly differs from the one commonly seen in  provenance literature where a ``witness'' refers to a subset of the input database records that is sufficient to ensure that a given output tuple appears in the result of a query \cite{DBLP:journals/ftdb/CheneyCT09}.} 
The set of witnesses of $D\models \exists \vec x\,(A_1 \wedge \ldots \wedge A_m)$ is the set
$\bigset{\vec w}{D \models (A_1 \wedge \ldots \wedge A_m)[\vec w /\vec x]}$.

A database instance may contain some ``forbidden'' tuples that may not be deleted. Since we are interested in the
data complexity of resilience, we specify \emph{at the query level} which tables contain tuples that
may or may not be deleted.  Those atoms from which tuples may not be deleted are called \emph{exogenous}\footnote{In other words, tuples in these atoms provide context and
  are outside the scope of possible ``interventions'' in the spirit of
  causality~\cite{HalpernPearl:Cause2005}.} and
we write these atoms or relations with a superscript ``$\exSymb$''. 
The other atoms, whose tuples may be deleted,
are called \emph{endogenous}.
We may occasionally attach the superscript ``$\enSymb$'' 
to an atom to emphasize that it is endogenous. Moreover, we can refer to a database
as a partition of its tables into its exogenous and endogenous parts, $D = D^\exSymb \cup D^\enSymb$.

\subsection{Query resilience}  
In this paper, we focus on determining the \emph{resilience} of a query with regard to changes in $D^\enSymb$.
Given $D\models q$, our motivating question is: 
what is the minimum number of tuples to remove in order to make the query false? 
\begin{definition}[Resilience]\label{def: resilience}
Given a query $q$ and database $D$, we say that $(D, k) \in \res(q)$ if and
only if $D \models q$ and there exists some $\Gamma \subseteq D^\enSymb$ such that $D -
\Gamma \not\models q$ and $|\Gamma| \leq k$.
\end{definition}
In other words, $(D,k) \in \res(q)$ means that there is a set of $k$ or fewer tuples in
the endogenous tables of $D$, the removal of which makes the query false. Observe that since $q$ is computable in \ptime, $\res(q)\in\NP$.  We will see that there is a dichotomy for
all sj-free conjunctive queries:  for all such queries $q$, either $\res(q)\in \PTIME$ or $\res(q)$ is \NP-complete (\autoref{resilience dichotomy thm}).
We are naturally interested in the optimization version of this decision problem:
given $q$ and $D$, find the \emph{minimum} $k$ so that $(D,k) \in \res(q)$. A
larger $k$ implies that the query is more ``\emph{resilient}'' and requires
the deletion of more tuples to change the query output. 

In this paper, we focus on Boolean queries, however we can also define the resilience problem for 
non-Boolean queries as follows:

\begin{definition}[Resilience for non-Boolean queries]\label{def: resilience non-boolean}
Given non-Boolean query $q(\vec y)$ and database $D$, we say that $(D, k) \in \res(q(\vec y))$ if and
only if $q(\vec y)^D \neq \emptyset$ and there exists some $\Gamma \subseteq D^\enSymb$ such that $q(\vec y)^{D - \Gamma} = \emptyset$ and $|\Gamma| \leq k$.
\end{definition}

\noindent
It is clear from the definition that we are interested in eliminating all the output tuples from the
query result, and it is easy to see that $\res(q(\vec y)) \equiv \res(q')$, where $q'$ is obtained
by removing all variables $\vec y$ from the head of $q$, turning them 
 into existential variables. 

We can refine this definition to include a target tuple $t$, i.e., instead of deleting \emph{all} output tuples from the query result, we would like 
to delete only \emph{one} output tuple $t$. 
As we saw 	in the introduction, this is the exact definition of the deletion propagation problem. The next subsection will make the correspondence between resilience and deletion propagation with source side-effects precise.

\subsection{Deletion propagation: source side-effects}

Deletion propagation in view updates generally refers to non-Boolean queries
$q(\vec y) \datarule {A}_1, \ldots, {A}_m$.
We next define the
problem~\cite{Buneman:2002,Dayal82} formally in our notation:

\begin{definition}[Source side-effects]
    Given a query $q(\vec y)$, database $D$, and an output tuple $t$, 
we say that $(D, t, k)\in\source(q(\vec y))$ if and
    only if $t \in q(\vec y)^D$ and there exists some $\Gamma \subseteq D$ such that
    $t \not \in q(\vec y)^{D - \Gamma}$ and $|\Gamma| \leq k$.
\end{definition}

It is easy to see that there is a homomorphism between resilience and the
source-side effect variant of deletion propagation. We have illustrated this
correspondence in \specificref{Example}{ex:introResilience} 
and next describe this transformation more formally. 

Given a conjunctive query 
$q(\vec y) \datarule {A}_1, \ldots, {A}_m$ and a tuple $t = \vec c$ in the output $q(\vec y)^D$. 
We first obtain a Boolean query $q'$ by deleting the head variables in $q(\vec y)$. Then we modify the database by applying a filter (selection):
for each relation $R_i(\vec z_i)$ we define a new relation $R_i'(\vec x_i) \datarule R_i(\theta_t(\vec z_i))$ with $\vec x_i$ being the existential variables that occur in $R_i$,
and where the substitution $\theta_{t}: \vec y \rightarrow \vec c$ replaces the former head variables with the corresponding constants from $t$ and keep the existential variables as they are.
For example, $R'(y) \datarule R(1,y)$ in \specificref{Example}{ex:introResilience} 
(see \autoref{Fig_SourceSideEffect} and \autoref{Fig_Resilience}).
This will lead to a new database $D' = \bigcup_i R_i'$
and a new Boolean query $q' \datarule {A}_1', \ldots, {A}_m'$, where $A_i' = R_i'(\vec x_i)$ if $A_i = R_i(\vec z_i)$,
for which the following holds:\footnote{An informal way to describe this transformation of $D$ at the query level is to first only keep tuples in the lineage of $t$ and to then delete all columns in atoms that contain constants from $\vec c$).}

\begin{corollary}[Resilience \& Source side-effects]
	Given a query $q(\vec y)$, database $D$, and output tuple $t \in q(\vec y)^D$, let $q'$ and $D'$ be the new Boolean query and new database instance obtained by the above transformation.
	Then: $(D, t, k)\in\source(q(\vec y)) \Leftrightarrow (D',k) \in \res(q')$.
\end{corollary}

Notice that the same transformation can be used to treat \emph{constants in a CQ} 
when considering source side-effects. 
Thus, by solving the complexity of resilience, we immediately also solve the problem of deletion propagation with source side-effects. 
We prefer to present our results using the notion
of resilience, as there are several applications beyond view updates that
relate to these problems. Examples include robustness of network connectivity
(identifying sets of nodes and edges that could disconnect a network),
deriving explanations for query results (finding the lineage tuples that have
most impact to an output), and problems related to set cover. We proceed to
discuss existing results on the complexity of deletion propagation with source
side-effects, and explain how our results on the complexity of resilience
extend this prior work.

\citeN{Buneman:2002} define a dichotomy for the hardness of
$\source(q)$ based only on the operations that occur in $q$, namely, selection, projection, join, union. 
Specifically, they show that
$\source(q(\vec y))$ is \NP-complete for PJ and JU queries (i.e., queries
involving projections and joins, or queries involving joins and unions), while
it is \ptime\ for SJ and SPU queries (i.e., queries involving selections and
joins, or queries involving selections, projections, and unions only). 
Later, \citeN{Cong12} 
showed that $\source(q(\vec y))$ is in \ptime\ for a SPJ query 
if all primary keys of the involved relations
appear in the head variables $\vec y$ (a condition called ``key preservation'').
Notice that the concept of key preservation does
not apply to the problem of resilience, as keys are never preserved in Boolean
queries.

In this paper, we identify a larger class of SPJ queries for which the problem
of resilience --- and thus $\source(q(\vec y))$ --- is in \ptime, thus extending
all prior results. 
In \Autoref{sec:resilience}, 
we provide a dichotomy result based on identifying a
specific and very intuitive structure in a query, called a \emph{triad}:
queries that contain a triad are \NP-complete, whereas those that do not are
in \ptime. 
Our results refine the prior work in the sense that prior results
characterize the dichotomy at the level of operators used in the query (e.g., joins,
projections), while our result 
identifies all polynomial cases based on ($i$) the actual
query and ($ii$) additional schema knowledge of forbidden, ``exogenous'' tables. 
In \Autoref{sec:FDs}, we extend our results to even include ($iii$) functional dependencies.

\subsection{Deletion propagation: view side-effects}
The problem of deletion propagation with view side-effects has a different objective than resilience: it attempts to minimize the changes in the view rather than the source.

\begin{definition}[View side-effects]\label{def:view}
    \looseness -1
    Given a query $q(\vec y)$, a database $D$, and a tuple $t$ in the
    view, we say that $(D, t, k)\in\view(q(\vec y))$ if and only if
    $t \in q(\vec y)^D$ and there exists some 
	$\Gamma \subseteq D$ such that 
	$t \not\in q(\vec y)^{D - \Gamma}$, and $|\Delta| \leq k$, where $\Delta = (q(\vec y)^D - (q(\vec y)^{D-\Gamma} \cup \set{t}))$.
	In other words, $\Delta$ is the set of tuples other than $t$ that were eliminated from the view.
\end{definition}

The dichotomy results from \citeN{Buneman:2002} extend to the
case of $\view(q)$, and the same is true for key
preservation~\cite{Cong12}. Later, \citeN{KimelfeldVW12}
refined the dichotomy for the view side-effect problem by providing a characterization that uses the query structure:
$\view(q(\vec y))$ is \ptime\ for queries that are \emph{head dominated}, and
\NP-complete otherwise. 
Head domination checks for the components of the query that are connected by the
existential variables, where all head variables contained in the atoms of that component appear in a
single atom in the query.
Our work in this paper offers a similar refinement for
the dichotomy of $\source(q(\vec y))$ from the characterization at the
operator level to the characterization at the level of query structure, plus knowledge of exogenous (``forbidden'') tables.

\smallskip
\introparagraph{Functional dependencies}
\citeN{Kimelfeld12} augmented the dichotomy on $\view(q)$ for
cases where functional dependencies (FDs) hold over the data instance $D$. The
tractability condition for this case checks whether the query has
\emph{functional head domination}, which is an extension of the notion of head
domination. We provide similar extensions in this paper for the problem of
$\source(q(\vec y))$: our dichotomy for the case of FDs checks for triads
after the query is structurally manipulated through a process we call
\emph{induced re\-writes}, which is basically a chase of FDs.

\smallskip
\introparagraph{Multi-tuple deletion}
\citeN{Cong12} also studied a variant of deletion propagation that
aims to remove a group of tuples from the view. Their results classify all
conjunctive queries as \NP-complete, but recently, 
\citeN{Kimelfeld:2013} provided a trichotomy for the class of sj-free CQs that
extends the notion of head domination, classifying queries into \ptime,
$k$-approximable in \ptime, and \NP-complete.

\subsection{Causal responsibility}\label{sec: responsibility}
A tuple $t$ is a \emph{counterfactual cause} for a query if by removing it the query changes from \true to \false. A tuple $t$ is an \emph{actual cause} if there exists a set $\Gamma$, 
called the \emph{contingency set}, removing of which makes $t$ a counterfactual cause. 
Determining actual causality is \NP-complete for general formulas~\cite{DBLP:journals/ai/EiterL02}, but there are families of tractable cases~\cite{EiterL06}.  Specifically, causality is \PTIME for all conjunctive queries~\cite{MeliouGMS11}.
Responsibility measures the \emph{degree} of causal contribution of a
particular tuple $t$ to the output of a query as a function of the size of a minimum contingency set:
$\rho=\frac{1}{1+\min{\Gamma}}$. 
These definitions stem from the work of \citeN{HalpernPearl:Cause2005}, and \citeN{ChocklerH04}, and were adapted to queries in previous work~\cite{MeliouGMS11}.
Even though responsibility ($\rho$) was originally
defined as inversely proportional to the size of the
contingency set $\Gamma$, here we alter this definition slightly to draw
parallels to the problem of resilience.

\begin{definition}[Responsibility]\label{def: responsibility}
Given query $q$, we say that $(D, t, k) \in \rsp(q)$ if and only if $D \models
q$ and there is $\Gamma \subseteq D^\enSymb$ such that $D - \Gamma \models
q$ and $|\Gamma| \leq k$ but $D - (\Gamma \cup \{ t\}) \not\models q$.
\end{definition}

In contrast to resilience, the problem of responsibility is defined for \emph{a
particular tuple} $t$ in $D$, and instead of finding a $\Gamma$ that will leave no witnesses 
for $D - \Gamma \models q$, we want to preserve only witnesses that involve $t$, so that
there is no witness left for $D - (\Gamma \cup \{ t\}) \models q$.
This difference, while subtle, is significant,
and can lead to different results. In \specificref{Example}{ex:introResilience}, the
resilience of query $q'$ has size 1 and contains tuple $t_1$. However, the
solution to the responsibility problem depends on the chosen tuple: the
contingency set of $s_1$ has size 2, and this size can be made arbitrarily
bigger by adding more tuples in $S$ with attribute $W=7$. Furthermore, we show that the problems differ in terms of their complexity.

For completeness, we briefly recall the notions of \emph{reduction}  and \emph{equivalence}
 in complexity theory: 

\begin{definition}[Reduction ($\leq$) and Equivalence ($\equiv$)]\label{reduction def}
For two decision problems, $S,T\subseteq \set{0,1}^*$,  we say that $S$ is \emph{reducible} to $T$
($S\leq T$) if there is an \emph{easy to compute} \emph{reduction} $f: \set{0,1}^*\rightarrow \set{0,1}^*$
such that 
\[ \forall w \in \set{0,1}^* \big( w \in S  \Leftrightarrow f(w) \in T \big)\; .\]

The idea is that the complexity of $S$ is less than or equal to the complexity of $T$ because any
membership question for $S$ (i.e., whether $w\in S$) can be easily translated into an equivalent question
for $T$, (i.e., whether $f(w) \in T$).  ``Easy to compute'' can be taken as 
expressible in first-order logic\footnote{All reductions in this paper are first-order, i.e., when
  we write $S\leq T$ we mean $S\leq_{\textrm{fo}}T$.  First-order reductions are natural for the
  relational database setting and they are more restrictive than logspace reductions, which in turn
  are more restrictive than polynomial-time reductions
 ($S\leq_{\textrm{fo}}T \Rightarrow  S\leq_{\textrm{log}}T \Rightarrow S\leq_{\textrm{p}}T$)
  \cite{Neil-book}. 
}. 
We say that two problems have \emph{equivalent} complexity ($S\equiv T$) iff they are inter-reducible, i.e., $S\leq T$ and
$T\leq S$.
\end{definition}

The problem of calculating resilience can always be reduced to the problem of calculating responsibility.

\begin{lemma}[$\res \leq \rsp$]\label{lem:res_rsp}
For any query $q$, $\res(q) \leq \rsp(q)$, i.e., there is a reduction
from $\res(q)$ to $\rsp(q)$.  Thus, if $\res(q)$ is hard (i.e., \NP-complete) then so is $\rsp(q)$.
Equivalently, if $\rsp(q)$ is easy (i.e., $\PTIME$) then so is $\res(q)$.
\end{lemma}

\begin{proof}
Let $q\datarule \exists x_1, \ldots, x_s \, A_1(\vec{z_1})\land\cdots\land A_r(\vec{z_r})$.
The reduction from $\res(q)$ to $\rsp(q)$ is as follows: given $(D,k)$, we map it to
$(D',\vec{t_0},k)$ where $D'$ consists of the database $D$ together with unique new values $a_1,
\ldots a_s$ and the new tuples $A_1(\vec{z_1}[\vec a/\vec x]), \ldots, A_r(\vec{z_r}[\vec a/\vec x])$.
In other words, we enter a completely new witness $\vec a$ for $q$ that has no values in common with
the domain of $D$.  Let $\vec{t_0}=A_1(\vec{z_1}[\vec a/\vec x])$, i.e., the tuple of 
these new values from atom $A_1$.  It follows that the size of the minimal contingency set for $q$
in $D$ is the 
same as the size of the minimal contingency set for $q$ and $\vec{t_0}$ in $D'$.  Thus, as desired,
$(D,k) \in \res(q) \Leftrightarrow (D',\vec{t_0},k) \in \rsp(q)$.
\end{proof}

Later we will see a query, $\rats$, for which $\res(\rats)\in\PTIME$ (\specificref{Cor.\@}{res(rats) easy}) but
$\rsp(\rats)$ is \NP-complete (\specificref{Prop.}{responsibility of rats is hard}). Thus (assuming
$\mathsf{P}\ne\mathsf{NP}$), $\rsp(q)$ is sometimes strictly harder than $\res(q)$.

%% file: full_3_Resilience.tex
\section{Complexity of resilience}\label{sec:resilience}
In this section we study the data complexity of resilience.  We prove that the  complexity of
resilience of a query $q$ can be exactly characterized
via a natural property of its \emph{dual hypergraph} $\mathcal{H}(q)$
(\specificref{Definition}{def:dualHypergraph}).
In \Autoref{sec: hard part}, we begin by showing that the resilience problem for two basic
queries, the triangle query ($q_\triangle$) and the tripod query ($q_\Tri$) are both \NP-complete.  
We then generalize these queries to a feature of hypergraphs that we call a \emph{triad} (\specificref{Definition}{def: triad}), which is a set of 3 atoms that are connected in a special way in $\mathcal{H}(q)$. 
We then prove that if $\mathcal{H}(q)$ contains a triad, then $\res(q)$ is \NP-complete, i.e., determining resilience is hard.  
Conversely, we show in \Autoref{sec: easy part} that if $\mathcal{H}(q)$ does not contain any triad, then $\res(q)\in$ \ptime.
We prove this by showing how to transform a triad-free sj-free CQ into a linear query $q'$ of equivalent complexity.  
The resilience of linear queries can be computed efficiently in polynomial time using 
a reduction to network flow as shown in previous  work~\cite{MeliouGMS11}.  
The desired dichotomy theorem for the resilience of sj-free CQ thus follows (\Autoref{resilience dichotomy thm}).

\subsection{Triads make resilience hard}\label{sec: hard part}
We will define triples of atoms called \emph{triads} and then prove that
if the dual hypergraph of a query $q$ contains a triad, then the resilience problem $\res(q)$ is
\NP-complete.  

We first define the (dual) hypergraph $\mathcal{H}(q)$ of query $q$. 
The hypergraph of a query $q$ is usually
defined with its vertices being the variables of $q$ and the hyperedges being the atoms \cite{abiteboul-hull-vianu}.
In this paper we use only the dual hypergraph:

\begin{definition}[Dual Hypergraph $\mathcal{H}(q)$]\label{def:dualHypergraph}
Let $q \datarule A_1,$ $\ldots, A_m$ be an sj-free CQ. 
Its \emph{dual hypergraph}  $\mathcal{H}(q)$
has vertex set $V=\{A_1,\ldots,A_m\}$.  Each variable $x_i \in \var(q)$ determines the hyperedge
consisting of all those atoms in which $x_i$ occurs: $\;e_i=\{A_j\,|\,x_i\in \var(A_j)\}$. 
\end{definition}

For example, \autoref{Fig_HardQueries} shows the dual hypergraphs of four important queries defined in
\specificref{Example}{four main queries ex}.
In this paper we only consider dual hypergraphs, so we use
the shorter term ``hypergraph'' from now on. In fact we will think of a query
and its hypergraph as one and the same thing.  Furthermore, when we discuss \emph{vertices},
\emph{edges} and \emph{paths}, we are referring to those objects in the hypergraph of the query
under consideration.  Thus, a \emph{vertex} is an
\emph{atom}, an  \emph{edge} is a \emph{variable}, and a \emph{path} is an alternating sequence of
vertices and edges, $A_1,x_1,A_2, x_2, \ldots, A_{n-1},x_{n-1},A_n$, such that for all $i$, 
$x_i \in \var(A_i)\cap \var(A_{i+1})$, i.e., the hyperedge $x_i$ joins vertices $A_i$ and
$A_{i+1}$.  
We explicitly list the hyperedges in the path, because more than one hyperedge may join
the same pair of vertices.
Furthermore, since disconnected
components of a query have no effect on each other, 
each of several disconnected components can be considered independently.
We will thus assume throughout that \emph{all queries
  are connected.} 
Similarly, without loss of generality, we assume no query contains two atoms with exactly the same set of 
variables.\footnote{If two atoms $A,B$ appear in $q$ with the identical
set of variables, we can replace $A$ by $A\cap B$ and delete $B$.}

\begin{example}[Important queries]\label{four main queries ex}
Before we precisely define what a triad is, we identify two hard queries,
$q_\triangle,q_\Tri$ and two related queries, $\rats,\brats$ (see
\autoref{Fig_HardQueries} for drawings of their hypergraphs).
\[\begin{array}{rcl@{\quad}l}
q_\triangle 	&\datarule& R(x,y),S(y,z),T(z,x)			&\textrm{(Triangle)}\\
\rats 			& \datarule& A(x),R(x,y),S(y,z),T(z,x) 		&\textrm{(Rats)}\\
\brats 			& \datarule& A(x),R(x,y),B(y),S(y,z),T(z,x) &\textrm{(Brats)}\\
q_\Tri	 		&\datarule& A(x),B(y),C(z), W(x,y,z)		&\textrm{(Tripod)}
\end{array}\]
\end{example}

We now prove that $q_\triangle$ and $q_\Tri$ are both hard, i.e., their resilience problems are \NP-complete. 
This will lead us to the definition of a triad: the hypergraph property that implies
hardness.   Later we will see that $\brats$ is easy for both
resilience and responsibility. \emph{However, counter to our initial intuition, $\rats$ is easy for resilience but hard for responsibility}.

\begin{figure}
 \definecolor{dg}{cmyk}{0.60,0,0.88,0.27}
\phantom{a}\hfill
\subfloat[\small{Triangle query $q_\triangle$}]{
\centering
	\begin{tikzpicture}[ scale=.13]
	\draw[color=red] (-8,6.93) node {{\color{black} $x$}};  
	\draw[color=dg] (8,6.93) node {{\color{black}  $y$}};  
	\draw[color=blue] (0,-3.5) node {{\color{black}  $z$}};  
	\draw (0,13.86) node {{\color{red} $R$}};  
	\draw (-8,0) node {{\color{blue} $T$}};  
	\draw (8,0) node {{\color{dg} $S$}};  
	\draw[rotate=-30] (-6.93,4) ellipse (2 and 12);  
	\draw[rotate=30] (6.93,4) ellipse (2 and 12);  
	\draw[rotate=90] (0,0) ellipse (2 and 12);  
	\end{tikzpicture}
\label{fig:triangleHypergraph}
}
\hfill
\subfloat[\small{Rats query $\rats$}]{
\centering
	\begin{tikzpicture}[ scale=.13]
	\draw (-8,6.93) node {$x$};  
	\draw (8,6.93) node {$y$};  
	\draw (0,-3.5) node {$z$};  
	\draw (-4.5,6.93) node {{\color{red} $A$}};  
	\draw (0,13.86) node {{\color{red} $R$}};  
	\draw (-8,0) node {{\color{blue} $T$}};  
	\draw (8,0) node {{\color{dg} $S$}};  
	\draw[rotate=-30] (-6.93,4) ellipse (2 and 12);  
	\draw[rotate=30] (6.93,4) ellipse (2 and 12);  
	\draw[rotate=90] (0,0) ellipse (2 and 12);  
	\end{tikzpicture}
\label{fig:ratsHypergraph}	
}
\hfill\phantom{a}

\phantom{a}\hfill
\subfloat[\small{Brats query $\brats$}]{
 \centering
	\begin{tikzpicture}[ scale=.13]
	\draw (-8,6.93) node {$x$};  
	\draw (8,6.93) node {$y$};  
	\draw (0,-3.5) node {$z$};  
	\draw (-4.5,6.93) node {{\color{red} $A$}};  
	\draw (4.5,6.93) node {{\color{dg} $B$}};  
	\draw (0,13.86) node {{\color{red} $R$}};  
	\draw (-8,0) node {{\color{blue} $T$}};  
	\draw (8,0) node {{\color{dg} $S$}};  
	\draw[rotate=-30] (-6.93,4) ellipse (2 and 12);  
	\draw[rotate=30] (6.93,4) ellipse (2 and 12);  
	\draw[rotate=90] (0,0) ellipse (2 and 12);  
	\end{tikzpicture}
\label{fig:bratsHypergraph}
}
\hfill
\subfloat[\small{Tripod query $q_\Tri$}]{
\vspace*{4ex}
 \centering
	\begin{tikzpicture}[ scale=.13]
	\draw (-8,6.93) node {$x$};  
	\draw (8,6.93) node {$y$};  
	\draw (9,10.5) node {$z$};  
	\draw (-4.5,6.93) node {{\color{red} $A$}};  
	\draw (10,14) node {{\color{blue} $C$}};  
	\draw (4.5,6.93) node {{\color{dg} $B$}};  
	\draw (0,13.8) node {{\color{black} $W$}};  
	\draw[rotate=-30] (-6.93,4) ellipse (2 and 12);  
	\draw[rotate=30] (6.93,4) ellipse (2 and 12);  
	\draw[rotate=90] (14,-6) ellipse (2 and 12);  
	\end{tikzpicture}
\vspace*{-1ex}
\label{fig:tripodHypergraph} 
}
\hfill\phantom{a}
\caption{\specificref{Example}{four main queries ex}: The hypergraphs of queries 
$q_\triangle$, $\rats$, $\brats$, $q_\Tri$. $\set{R,
S,T}$ is a triad of $q_\triangle$; \set{A,B,C} is a triad of $q_\Tri$.
}\label{Fig_HardQueries}
\end{figure}

\begin{proposition}[Triangle $q_\triangle$ is hard]\label{thm: hardness of triangle}
$\res(q_\triangle)$ and $\rsp(q_\triangle)$ are \NP-complete.
\end{proposition}

\begin{proof}
We reduce 3SAT to $\res(q_\triangle)$.  It will then follow that $\res(q_\triangle)$ is \NP-complete,
and thus so is $\rsp(q_\triangle)$ by \specificref{Lemma}{lem:res_rsp}.
Let $\psi$ be a 3CNF formula with $n$ variables $v_1, \ldots, v_n$ and $m$ clauses 
$C_0, \ldots, C_{m-1}$.  
Our reduction will map  any such $\psi$ to a pair $(D_\psi,k_\psi)$ where $D_\psi$ is a database 
satisfying
$q_\triangle$, and 

\myequation{hardRatsReduction}{\psi\in 3\sat \quad\Leftrightarrow\quad (D_\psi,k_\psi) \in \res(q)}

In our construction, if $\psi \in 3\sat$, then the size of each minimum contingency set for $q_\triangle$ in
$D_\psi$ will be $k_\psi=6mn$, whereas if $\psi \not\in 3\sat$, then the size of all contingency sets
for $q_\triangle$ in $D_\psi$ will be greater than $k_\psi$.

Notice that $D_\psi \models q_\triangle$ iff it contains three tuples $R(a,b)$, $S(b,c)$, $T(c,a)$ that together form a witness.
We visualize $R(a,b)$ as a red edge, $S(b,c)$ as a green edge and $T(c,a)$ as a
blue edge.
In other words, each witness $(a,b,c)$ for $D_\psi\models q_\triangle$ forms an RGB triangle.
(Notice that the edge direction $a \rightarrow b$ drawn in 
\specificref{Figures}{fig:gi segment}, \ref{fig:gadget} and
\ref{fig:gadgetgi identification}
corresponds to the
variable order in $R$, and analogously for $S$ and $T$.)
The job of a contingency set for
$q_\triangle$ is to remove all RGB triangles.

$D_\psi$ contains one circular gadget $G_i$ for each variable $v_i$.  
The circle consists of $12m$ solid edges, half of
them marked $v_i$ and the other half marked $\ov{v_i}$ (see \specificref{Figures}{fig:gi segment}, \ref{fig:gadget}).  Note that there are $12m$ RGB triangles and they can be minimally broken by
choosing the $6m$ $v_i$ edges or the $6m$ $\ov{v_i}$ edges. Any other way would require more
edges removed.
Thus, each minimum contingency set for $D_\psi$ corresponds to a truth assignment to the
variables of $\psi$. And there will be
a minimum contingency set of size $k_\psi = 6mn$ iff $\psi \in 3\sat$.

\begin{figure}
\begin{center}
\begin{tikzpicture}[scale=.35,
			every circle node/.style={fill=white, minimum size=7mm, inner sep=0,draw}]
\node (11) at (2,2)  [circle] {$a^i_1$};
\node (12) at (6,2)  [circle] {$b^i_1$};
\node (13) at (10,2) [circle] {$c^i_1$};
\node (14) at (14,2) [circle] {$a^i_2$};
\node (15) at (18,2) [circle] {$b^i_2$};
\node (16) at (22,2) [circle] {$c^i_2$};
\node (22) at (25,2) [] {};

\path[->, line width=2pt, auto] 
	(11) 	edge[color=red] 	node {${v_i}$} (12) 
	(12) 	edge[color=dg] 		node {$\ov{v_i}$} (13)
	(13) 	edge[color=blue] 	node {${v_i}$} (14)
	(14) 	edge[color=red] 	node {$\ov{v_i}$} (15)
	(15) 	edge[color=dg] 		node {${v_i}$} (16)
	(16) 	edge[color=blue]	node {$\ov{v_i}$} (22);	
	
\path[->, line width=2pt, dotted, out=-130, in=-50] 		
	(13) edge[color=blue] 	node [-, solid]{ $\Sadey[1.5][white]$}		(11)
	(14) edge[color=red] 	node [-, solid]{ $\Sadey[1.5][white]$}		(12)
	(15) edge[color=dg] 	node [-, solid]{ $\Sadey[1.5][white]$}		(13)
	(16) edge[color=blue] 	node [-, solid]{ $\Sadey[1.5][white]$}		(14);
\end{tikzpicture}
\end{center}
\caption{A six-node segment of the gadget $G_i$ in the hardness proof for $q_\triangle$: A minimum contingency set chooses either all the solid lines marked $v_i$, or all the solid lines marked $\ov{v_i}$.  The dotted lines are sad because each of them is only part of one single RGB triangle, thus they are never chosen.}
\label{fig:gi segment}
\end{figure}
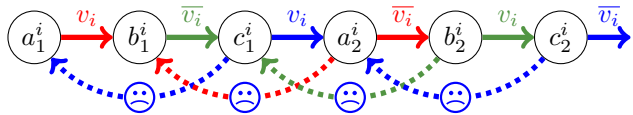

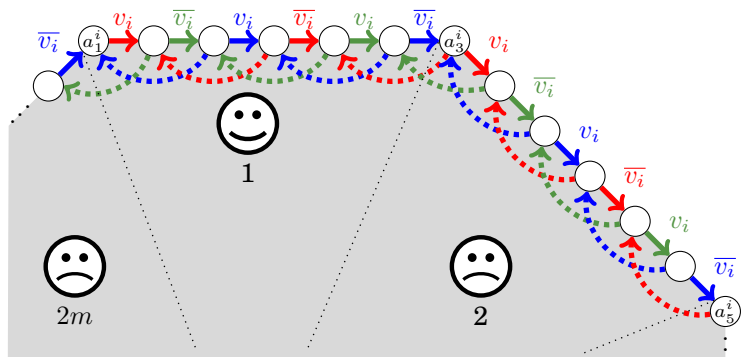
\begin{figure}
\begin{center}
\begin{tikzpicture}[
	scale=.2,
	every circle node/.style={fill=white, minimum size = 4mm, draw},
	]

\node (05) at (-2.5,-2.5) [inner sep=0mm] {\rotatebox{-45}{$\vdots$}};
\node (06) at (-1,-1) [circle] {};
\node (11) at (2,2)  [circle] {};
\node at (11) 	{\scriptsize$a_1^i$};
\node (12) at (6,2)  [circle] {};
\node (13) at (10,2) [circle] {};
\node (14) at (14,2) [circle] {};
\node (15) at (18,2) [circle] {};
\node (16) at (22,2) [circle] {};
\node (21) at (26,2) 	[circle] {};
\node at (21) 	{\scriptsize$a_3^i$};
\node (22) at (29,-1)  	[circle] {};
\node (23) at (32,-4)  	[circle] {};
\node (24) at (35,-7)  	[circle] {};
\node (25) at (38,-10)  [circle] {};
\node (26) at (41,-13) 	[circle] {};
\node (31) at (44,-16)  [circle] {};
\node at (31) 	{\scriptsize$a_5^i$};
\node (32) at (44,-18)  [inner sep=0mm] {$\vdots$};

\path[->, line width=2pt, auto] 
	(06) 	edge[color=blue]	node {$\ov{v_i}$} (11)
	(11) 	edge[color=red] 	node {${v_i}$} (12) 
	(12) 	edge[color=dg] 		node {$\ov{v_i}$} (13)
	(13) 	edge[color=blue] 	node {${v_i}$} (14)
	(14) 	edge[color=red] 	node {$\ov{v_i}$} (15)
	(15) 	edge[color=dg] 		node {${v_i}$} (16)
	(16) 	edge[color=blue] 	node {$\ov{v_i}$} (21)
	(21) 	edge[color=red] 	node {${v_i}$} (22)
	(22) 	edge[color=dg] 		node {$\ov{v_i}$} (23)
	(23) 	edge[color=blue] 	node {${v_i}$} (24)
	(24) 	edge[color=red] 	node {$\ov{v_i}$} (25)
	(25) 	edge[color=dg] 		node {${v_i}$} (26)
	(26) 	edge[color=blue]	node {$\ov{v_i}$} (31);	
\path[->, line width=2pt, dotted, out=-120, in=-60] 		
	(12) edge[color=dg, in=-30, out=-110] (06.east)		
	(13) edge[color=blue] (11)
	(14) edge[color=red] (12)
	(15) edge[color=dg] (13)
	(16) edge[color=blue] (14)
	(21) edge[color=red] (15);
\path[->, line width=2pt, dotted, out=-170, in=-100] 		
	(22) edge[color=dg, in = -60] (16)
	(23) edge[color=blue] (21)
	(24) edge[color=red] (22)
	(25) edge[color=dg] (23)
	(26) edge[color=blue] (24)
	(31) edge[color=red] (25);	

\node (corner) at (-3.7,-19.3) {};
\node (cent1) at (9,-19) {};
\node (cent2) at (16,-19) {};
\node (cent3) at (34,-19) {};
\node (cent0) at (4,-19) {};

\path[-, line width=0.5pt, dotted, ] 		
	(cent1) edge (11.south west)
	(cent2) edge (21.west)
	(cent3) edge (31.north west);

\begin{scope}[on background layer] 	
\fill [black!15] (32.south)--(31.center)--(21.center)--
				(11.center)--(05.south west)
		--(corner.center);
\end{scope}

\draw (12.5,-3.5) 	node {{\color{black}  $\Smiley[3][white]$ }};
\draw (12.5,-6.75) 	node {{\color{black}  1 }};
\draw (28,-13) 	node {{\color{black} $\Sadey[3][white]$ }};
\draw (28,-16.25) node {{\color{black}  2 }};
\draw (1,-13) 		node {{\color{black} $\Sadey[3][white]$ }};
\draw (1,-16.25) 	node {{\color{black} $2m$ }};

\end{tikzpicture}
\end{center}
\caption{Each gadget $G_i$ in the hardness proof for $q_\triangle$ is a cycle containing $2m$
  six-node segments and a total of $12m$ RGB triangles.  They can all be
  eliminated by removing the $6m$ edges marked $v_i$ or the $6m$ edges marked $\ov{v_i}$.  
The even numbered segments are sad because they are never used for connecting different gadgets (corresponding to clauses that use several variables);
they only separate the odd ones, thus preventing spurious triangles.
}
\label{fig:gadget}
\end{figure}

\begin{figure}
\begin{center}
\begin{tikzpicture}[
	scale=0.9,
	PT/.style = {inner sep=0mm},
	CIRC/.style = {circle split, minimum size=18mm, draw, fill=white}]
\node (A) at (14,2) 	[CIRC, rotate = -30] {};
\node at ($(A) + (+0.4,+0.4)$) 	{$a^1_{4j+1}$};
\node at ($(A) + (-0.2,-0.4)$)  {$a^3_{4j+2}$};

\node (B) at (18,2) 	[CIRC, rotate = 30] {};
\node at ($(B) + (-0.2,+0.4)$) 	{$b^1_{4j+1}$};
\node at ($(B) + (+0.25,-0.4)$) {$b^2_{4j+1}$};

\node (C) at (16,-.67) [CIRC, rotate = 90] {};
\node at ($(C) + (-0.47,0)$) 	{$c^3_{4j+1}$};
\node at ($(C) + (+0.50,0)$) 	{$c^2_{4j+1}$};

\node (G11) at 	(12.65,3.77) 	[PT] {};
\node (G12) at 	(19.35,3.77) 	[PT] {};
\node (G31) at 	(11.5,2) 		[PT] {};
\node (G32) at 	(14.7,-2.4) 	[PT] {};
\node (G33) at 	(11.5,-2.4) 	[PT] {};
\node (G21) at 	(20.5,2) 		[PT] {};
\node (G22) at 	(17.3,-2.4) 	[PT] {};
\node (G23) at 	(20.5,-2.4) 	[PT] {};

\path[->, line width=2pt, auto] 
	(G11) 	edge[color=blue] (A)
	(A) 	edge[color=red] node {$v_1$} (B)
			edge[color=red] (G31)
	(G21) 	edge[color=red] (B)
	(B) 	edge[color=dg] node {$\ov{v_2}$} (C)
			edge[color=dg] (G12)
	(G32) 	edge[color=dg] (C)
	(C) 	edge[color=blue] node {$v_3$} (A)
			edge[color=blue] (G22);

\begin{scope}[on background layer] 	
\fill [black!15] (B.center)--(A.center)--(G11.center)--(G12.center);
\fill [black!15] (A.center)--(C.center)--(G32.center)--(G31.center);
\fill [black!15] (B.center)--(C.center)--(G22.center)--(G21.center);
\end{scope}

\draw (16,3) 		node {$G_1$};
\draw (18.1,-.1) 	node {$G_2$}; 
\draw (13.9,-.1) 	node {$G_3$}; 

\end{tikzpicture}
\end{center}

\caption{For clause $C_j=(v_1 \lor \ov{v_2} \lor v_3)$ in the hardness proof for $q_\triangle$,  we identify vertices 
$b^1_{4j+1}\in G_1$ with $b^2_{4j+1}\in G_2$; 
$c^2_{4j+1}\in G_2$ with $c^3_{4j+1}\in G_3$ and 
$a^3_{4j+2}\in G_3$ with  $a^1_{4j+1}\in G_1$. 
This RGB triangle will be deleted iff the chosen variable
  assignment satisfies $C_j.$}
\label{fig:gadgetgi identification}
\end{figure}
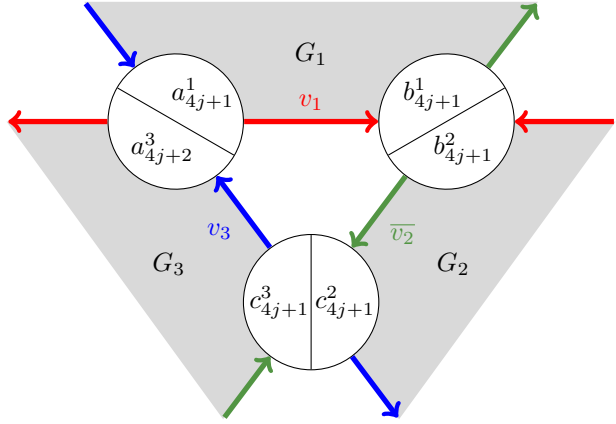

We complete the construction of $D_\psi$ by adding one RGB triangle
for each clause $C_j$.  For example, suppose
$C_j = v_1 \lor \ov{v_2} \lor v_3$.
The RGB triangle we add consists of a red edge marked $v_1$, a green edge marked $\ov{v_2}$ and
a blue edge marked $v_3$ (see \autoref{fig:gadgetgi identification}).  Note that if the chosen assignment satisfies $C_j$,  then all $v_1$ edges
are removed, or all $\ov{v_2}$ edges are removed, or all $v_3$ edges are removed.  Thus the $C_j$
triangle is automatically removed.  

How do we create $C_j$'s RGB triangle?  Remember that we have chosen $G_i$ to contain 2 segments for
each clause.  We use segment $2j+1$ of $G_i$ to produce the $v_i$ or
$\ov{v_i}$ used in $C_j$'s triangle.  
The even numbered segments are not used:  they serve as
buffers to prevent spurious RGB triangles from being created. In \autoref{fig:gadget}, we 
mark these even segments with frowns:  they are sad because they are never used.

More precisely, the  red
$v_1$-edge from $G_1$ is $(a^1_{4j+1},b^1_{4j+1})$, 
the green $\ov{v_2}$-edge from $G_2$ is $(b^2_{4j+1},c^2_{4j+1})$,
and the blue $v_3$-edge from $G_3$ is $(c^3_{4j+1},a^3_{4j+2})$ 
(see \autoref{fig:gadgetgi identification}).  

Now to make this an RGB triangle in $D_\psi$, we identify the two $a$-vertices, the two $b$ vertices and the
two $c$ vertices.  In other words, $G_1$'s $a$-vertex $a^1_{4j+1}$ is equal to $G_3$'s $a$-vertex
$a^3_{4j+2}$,
i.e., they are the same element of the domain of $D_\psi$.
We have thus constructed $C_j$'s RGB triangle  (see \autoref{fig:gadgetgi identification}).

The key idea is that these
identifications can only create this single new RGB triangle because there is no other way to get back to
$G_1$ from $G_2$ in two steps.  All other identifications involve
different segments and so are at least six steps away. Recall that this is the reason why the even-numbered
segments in the $G_i$'s are not used: this ensures that no spurious RGB triangles
are created.
Thus, as desired, \specificref{Eq.}{hardRatsReduction} holds and we have reduced $3\sat$ to $\res(q_\triangle)$.
\end{proof}

We next show that the tripod query $q_\Tri$ is also hard.  We do this by reducing the triangle to the
tripod.  
Understanding this reduction  is useful for understanding the proof of our main result.

\begin{proposition}[Tripod $q_\Tri$ is hard]\label{prop:tripodQuery}
$\res(q_\Tri)$ and $\rsp(q_\Tri)$ are \NP-complete.
\end{proposition}

\begin{proof}
First observe that in $q_\Tri$, $\var(A)$ is a subset of $\var(W)$.  We say that
$A$ \emph{dominates} $W$ (\specificref{Definition}{domination}).  It thus follows that when computing the resilience
of $q_\Tri$, a tuple $W(a,b,c)$ is never needed in a minimum contingency set because it could always
be replaced at least as efficiently by the tuple $A(a)$.  It follows that we may assume that $W$ is
exogenous, i.e., $\res(q_\Tri) \equiv
\res(q'_\Tri)$ where $q'_\Tri \datarule A(x),B(y),C(z),W^\exSymb(x,y,z)$ (\specificref{Prop.}{fact: domination does not change complexity}).

We now reduce $\res(q_\triangle)$ to $\res(q'_\Tri)$.  It will then follow that $\res(q_\Tri)$ is \NP-complete,
and thus so is $\rsp(q_\Tri)$ by \specificref{Lemma}{lem:res_rsp}.
Let $(D,k)$ be an instance of $\res(q_\triangle)$.  We construct an instance $(D',k)$ of $\res(q'_\Tri)$
by constructing relations $A,B,C$ as copies of $R,S,T$ from $D$.  Define $D'=(A,B,C,W^\exSymb)$ as follows:
\begin{align*}
A &= \bigset{\angle{ab}}{R(a,b)\in D}\\
B &= \bigset{\angle{bc}}{S(b,c)\in D}\\
C &= \bigset{\angle{ca}}{T(c,a)\in D}\\
W^\exSymb &= \bigset{(\angle{ab},\angle{bc},\angle{ca})}{a,b,c\in \textrm{dom}(D)}
\end{align*}

\noindent Here, $\dom(D)$ is the set of domain elements of $D$ and $\angle{ab}$ stands for a new
unique domain value resulting from the concatenation of domain values $a$ and $b$.

Observe that there is a 1:1 correspondence between the witnesses of $D\models
q_\triangle$ and the witnesses of $D'\models q'_\Tri$. For example, $(a,b,c)$ is a witness that 
$D\models q_\triangle$ iff tuples $R(a,b),S(b,c),T(c,a)$ occur in $D$.  This holds iff
$(\angle{ab},\angle{bc},\angle{ca})$ is a witness that $D'\models 
q'_\Tri$, i.e., the tuples
$A(\angle{ab}),B(\angle{bc}),C(\angle{ca}),W(\angle{ab},\angle{bc},\angle{ca})$ occur in $D'$.
Thus, every contingency set for $q_\triangle$ in $D$ corresponds to a contingency set of the same
size for $q'_\Tri$ in $D'$.
It follows that $(D,k) \in \res(q_\triangle) \Leftrightarrow (D',k)\in \res(q'_\Tri)$.
\end{proof}

While $q_\triangle$ and $q_\Tri$ appear to be very different, they
share a key common structural property, which we define next.

\begin{definition}[triad]\label{def: triad}
A \emph{triad} is a set of three endogenous atoms, ${\cal T} = \set{S_0,S_1,S_2}$ 
such that for every pair $i,j$, 
there is a path from $S_i$ to $S_j$ that uses no variable occurring in the other 
atom of ${\cal T}$. 
\end{definition}

Observe that atoms $R,S,T$ form a triad in $q_\triangle$ and atoms $A,B,C$ form a triad in
$q_\Tri$ (see \autoref{Fig_HardQueries}).
For example, there is a path from $R$ to $S$ in $q_\triangle$ (across hyperedge $y$) that uses only variables (here $y$) that are not contained in the other atom (here $y \not \in \var(T)$).

A triad is composed of endogenous atoms. Some atoms such as $W$ in $q_T$ are given as endogenous,
but are not needed in contingency sets.  We will simplify the query by making all such atoms 
exogenous.  

\begin{definition}[Domination]\label{domination}
If a query $q$ has endogenous atoms $A,B$ such that $\var(A) \subset \var(B)$, then we say that $A$
\emph{dominates} $B$.\footnote{Recall that we never have the case of $\var(A) = \var(B)$.}
\end{definition}	

We already saw an example in \specificref{Prop.}{prop:tripodQuery}: in $q_\Tri$, each of the atoms $A,B,C$
dominates $W$. The following proposition was proved in \cite{MeliouGMS11}.  Unfortunately however,
it was claimed to hold with respect to \emph{responsibility} rather than \emph{resilience}.  As we will see later,
this proposition fails for responsibility because the tuple we are computing the responsibility of
may interfere with domination (\specificref{Prop.}{responsibility of rats is hard}).

\begin{proposition}[Domination for resilience]\label{fact: domination does not change complexity}
Let $q$ be an sj-free CQ and $q'$ the query resulting from labeling some dominated atoms as exogenous.
Then $\res(q) \equiv \res(q')$.
\end{proposition}

\begin{proof}
Let $\Gamma$ be a minimum contingency set of $q$ in $D$.  
Suppose that atom $A$ dominates atom $B$ but there is some tuple $B(\vec t)\in \Gamma$. 
Let $\vec p$ be the projection of $\vec t$ onto $\var(A)$. 
Then we
can replace $B(\vec t)$ by $A(\vec p)$ 
and we remove at least as many witnesses that $D\models q$.  It
follows, as desired, that the complexity of $\res(q)$ is unchanged if $B$ is exogenous, i.e.,  $\res(q) \equiv \res(q')$.
\end{proof}

When studying resilience, we follow the convention that 
\emph{all dominated atoms are exogenous}.  For example, 
$A$ dominates $R$ and $S$ in the query $\rats$, 
and $B$ dominates $R$ and $S$ in the query $\brats$.  We
thus transform the queries so that the dominated atoms are exogenous.  Exogenous atoms have the superscript ``$\exSymb$''.
\myequation{eqn:raxx}{\begin{array}{rcl}
\raxx & \datarule& A(x),R^\exSymb(x,y),S(y,z),T^\exSymb(z,x) \\
\brxxx & \datarule& A(x),R^\exSymb(x,y),B(y),S^\exSymb(y,z),T^\exSymb(z,x)
\end{array}}

\noindent
By \specificref{Prop.}{fact: domination does not change complexity}, $\res(\rats) \equiv \res(\raxx)$ and 
$\res(\brats) \equiv \res(\brxxx)$.

We now prove our first main result.

\begin{lemma}[Triads make $\res(q)$ hard]\label{hard part dichotomy}
Let $q$ be an sj-free CQ where all dominated atoms are exogenous. If $q$ has a triad, then $\res(q)$
is \NP-complete. 
\end{lemma}

\begin{proof}
Let $q$ be a query with triad ${\cal T}=\set{S_0,S_1,S_2}$.
We build a reduction from $\res(q_\triangle)$ to $\res(q)$. 
Given any $D$ that satisfies $q_\triangle$ we will produce a database $D'$ that satisfies $q$ such
that for all $k$:
\myequation{eq res hard case1}{(D,k)\in \res(q_\triangle) \quad\Leftrightarrow\quad (D',k)\in\res(q)}

\noindent
We will assume that no variable is shared by all three elements of ${\cal T}$  
(we can ignore any
such variable by setting it to a constant). 
Our proof splits into two cases:

\emph{Case 1}:  $\var(S_0), \var(S_1), \var(S_2)$ 
are pairwise disjoint: 
Our reduction is similar to the reduction from $q_\triangle$ to $q_\Tri$ (\specificref{Prop.}{prop:tripodQuery}).

We first define the triad relations in $D'$:
\myequation{case1Eq1}{
\begin{array}{rcl}
S_0 &=& \bigset{(\angle{ab}, \ldots, \angle{ab})}{R(a,b) \in D} \\
S_1 &=& \bigset{(\angle{bc}, \ldots, \angle{bc})}{S(b,c) \in D} \\
S_2 &=& \bigset{(\angle{ca}, \ldots, \angle{ca})}{T(c,a) \in D}. 
\end{array}}
Thus, each tuple of, for example, $S_0$ consists of identical entries with value $\angle{ab}$ for each
pair $R(a,b) \in D$.  Thus, $S_0,S_1,S_2$ mirror $R,S,T$, respectively. 

To define all the relations corresponding to the other atoms $A_i$ of $D'$, we first
partition the variables of $q$ into 4 disjoint sets: $\var(q)= \var(S_0)\cup\var(S_1)\cup \var(S_2) \cup
V_3$.  Now for each atom $A_i$, arrange its variables in these four groups. 
Then define the relation $R'_i$ of $D'$ corresponding to atom $A_i$ as follows
\myequation{case1 eq}{R_i' = \bigset{(\angle{ab};\angle{bc};\angle{ca};\angle{abc})}{D \models q_\triangle(a,b,c)}}
\noindent For example, all the variables $v\in \var(S_0)$ are assigned the value $\angle{ab}$
and all the variables $v\in V_3$ are assigned $\angle{abc}$.

By the definition of triad, there is a path from $S_0$ to $S_1$ not using any edges 
(variables) from
$\var(S_2)$. 
Thus, any witness of  $D'\models q$ 
that includes occurrences of $\angle{ab}$ and $\angle{b'c'}$ must have $b =
b'$.  

Similarly, a path from $S_1$ to $S_2$ guarantees that $c$ is preserved and a path from $S_2$ to
$S_0$ guarantees that $a$ is preserved.  It follows that the witnesses that $D' \models q$ are
essentially identical to the witnesses that $D\models q_\triangle(x,y,z)$ (see \autoref{hard part
  fig}).\footnote{More precisely, if $(a,b,c)$ is a witness that $D\models
  q_\triangle$, then $(\angle{ab},\angle{bc},\angle{ca},\angle{abc},a,b,c)$ is a witness that
  $D'\models q$, with the variables partitioned according to Eq. \ref{variable partition eq},
and these are the only
  possible such witnesses.  
}

Furthermore, any minimum contingency set only needs tuples from
$S_0, S_1$ or $S_2$.  Thus the sizes of minimum contingency sets are preserved, i.e., \specificref{Eq.}{eq res
  hard case1} holds, as desired.  Thus $\res(q)$ is \NP-complete.

\emph{Case 2}:  $\var(S_i) \cap \var(S_j) \ne \emptyset $  for some $i\ne j$:
We generalize the construction from Case 1 as follows.  Partition $\var(S_i)$ into those 
unshared, those shared with $S_{i-1}$, and those shared with $S_{i+1}$ (addition here is mod 3).

We then assign the
relations of the triad as follows:
\begin{align*}
S_0 &= \bigset{(\angle{ab}; a; b)}{R(a,b) \in D} \\
S_1 &= \bigset{(\angle{bc}; b; c)}{S(b,c) \in D} \\
S_2 &= \bigset{(\angle{ca}; c; a)}{T(c,a) \in D}
\end{align*}
Since none of the $S_i$'s is dominated, both $a$ and $b$ occur in each tuple of $S_0$, both of $b$
and $c$ in each tuple of $S_1$ and both of $c$ and $a$ in each tuple of $S_2$.
Thus, as in Case 1, $S_0,S_1,S_2$ capture $R,S,T$, respectively.  
The key ideas is now that we partition all the variables $\var(q)$ into 7 sets 
according to their respective appearance in each of the 3 tables.
For each assignment of $x,y,z$ to values $a,b,c$ in $D$, we will then make assignments to the variables according to their partition:
\begin{align}
\begin{tabular}{ >{$}c<{$} | >{$}c<{$} | >{$}c<{$} }\label{variable partition eq}
\textrm{set name} & \textrm{variable partition} 						& \textrm{assignment}	\\
 \hline
V_0 & \var(S_0)-(\var(S_1) \cup \var(S_2))				& \angle{ab}  	\\  
V_1 &  \var(S_1)-(\var(S_0) \cup \var(S_2))				& \angle{bc} 	\\
V_2 &  \var(S_2)-(\var(S_0) \cup \var(S_1))				& \angle{ca} 	\\
V_3 &  \var(q) - (\var(S_0)\cup \var(S_1)\cup \var(S_2))	& \angle{abc} 	\\              
V_4 &  \var(S_2) \cap \var(S_0)							& a 	\\
V_5 &  \var(S_0) \cap \var(S_1)							& b		\\
V_6 &  \var(S_1) \cap \var(S_2)							& c 	
\end{tabular}
\end{align}

We then define the relations in $D'$ corresponding to each of the other atoms $A$ of $q$
to be the following set of tuples, where the only difference is which
of the 7 members of the partition of variables occurs in $\var(A)$.
\myequation{case2 eq}{\hspace*{-.1in}\bigset{(\angle{ab};\angle{bc};\angle{ca};\angle{abc};a;b;c)}
	{D\!\models\! q_\triangle(a,b,c)}}

By the definition of a triad, there is a path from $S_0$ to $S_1$ not using any edges (variables) from
$S_2$.  Thus, ``$b$'' is always present (see Eq. \ref{variable partition eq}).
Thus, any witness including occurrences of
some of  $\angle{ab},b',\angle{b''c}$ must 
have $b = b' = b''$.  Thus, as in Case 1, 
the witnesses of $D' \models q$ are essentially identical to the
witnesses of $D\models q_\triangle$ and we have reduced
$\res(q_\triangle)$ to $\res(q)$ (see \autoref{hard part fig}).
\end{proof}

\begin{figure}
\begin{center}
\subfloat[Case 1]{
\begin{tikzpicture}[ scale=.3,
	circ/.style = {circle, minimum size=9mm, draw, fill=white}]]
\node[circ] (S0) at(0,0)  {{\color{red} $S_0(\angle{ab})$}};
\node[circ] (S1) at(6.5,-8)  {{\color{dg} $S_1(\angle{bc})$}};
\node[circ] (S2) at(-6.5,-8)  {{\color{blue} $S_2(\angle{ac})$}};
\draw (6.5,-2.5) node {\begin{tabular}{c}\\
                                $b$\\
                             preserved\\
                       \end{tabular} };
\draw (-6.5,-2.5) node {\begin{tabular}{c}\\
                                $a$\\
                             preserved\\
                       \end{tabular} };
\draw (0,-9) node {  $c$ preserved};
\path[line width=2pt] (S0) edge (S1);
\draw[line width=2pt] (S1) -- (S2);
\draw[line width=2pt] (S2) -- (S0);
\end{tikzpicture}
\label{hard part fig case 1}
}
\hspace*{.5in}
\subfloat[Case 2]{
\begin{tikzpicture}[ scale=.3,
	circ/.style = {circle, minimum size=9mm, draw, fill=white}]]
\node[circ] (S0) at(0,0)  {{\color{red} $S_0(\angle{ab};a;b)$}};
\node[circ] (S1) at(6.5,-8)  {{\color{dg} $S_1(\angle{bc};b;c)$}};
\node[circ] (S2) at(-6.5,-8)  {{\color{blue} $S_2(\angle{ca};c;a)$}};
\draw (6.5,-2.5) node {\begin{tabular}{c}\\
                                $b$\\
                             preserved\\
                       \end{tabular} };
\draw (-6.5,-2.5) node {\begin{tabular}{c}\\
                                $a$\\
                             preserved\\
                       \end{tabular} };
\draw (0,-9) node {\begin{tabular}{c}\\
                                $c$\\
                             preserved\\
                       \end{tabular} };
\path[line width=2pt] (S0) edge (S1);
\draw[line width=2pt] (S1) -- (S2);
\draw[line width=2pt] (S2) -- (S0);
\end{tikzpicture}
\label{hard part fig case 2}
}
\end{center}
\caption{Reduction from $\res(q_\triangle)$ to $\res(q)$ when $q$ contains a triad
  $\set{S_0,S_1,S_2}$ in the proof of 
\specificref{Lemma}{hard part dichotomy}.
}
\label{hard part fig}
\end{figure}
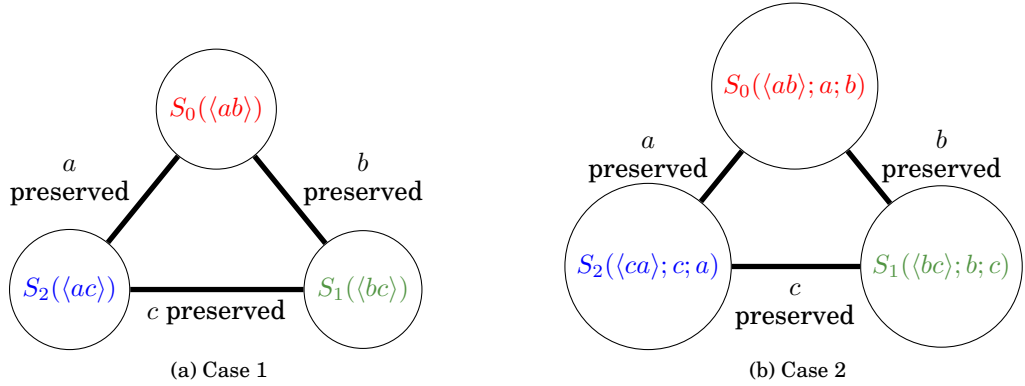

\subsection{Polynomial algorithm for linear queries}\label{sec: easy part}
We just showed that resilience for queries with triads is
\NP-complete. Next we will prove a strong converse: 
resilience for triad-free queries is in \ptime.
We start by defining a class of queries for which resilience is
known to be in \ptime.

\begin{definition}[Linear Query]\label{def: linearHypergraph}
A query $q$ is \emph{linear} if its atoms may be
arranged in a linear order such that each variable occurs in a contiguous sequence of atoms.
\end{definition}

\begin{example}[Linear Query]\label{linear ex}
Geometrically, a query is linear if all of the vertices of its hypergraph can be drawn along a straight
line and all of its hyperedges can be drawn as convex regions.
For example, the following query is linear: $q \datarule A(x),R(x,y),S(y,z)$   (see \autoref{linear fig}).
\end{example}

\begin{figure}
\begin{center}
\begin{tikzpicture}[ scale=.3]
\draw (-6,0) node {{\color{red} $A$}};
\draw (0,0) node {{\color{dg} $R$}};  
\draw (6,0) node {{\color{blue} $S$}};  
\draw (3,-2.2) node {{  $y$}};
\draw (7.3,0) node {{  $z$}};
\draw (-3,-2.2) node {{ $x$}};
\draw (-3,0) ellipse (6.5 and 1.5);
\draw (6,0) ellipse (1 and 1);
\draw (3,0) ellipse (6.5 and 1.5);
\end{tikzpicture}
\end{center}
\caption{\specificref{Example}{linear ex}: Linear query  $q \datarule A(x),R(x,y),S(y,z)$.
\label{linear fig}}
\end{figure}

The responsibility of linear queries is known to be in $\PTIME$ and thus by
\Autoref{lem:res_rsp}, resilience of linear queries is in $\PTIME$ as well.

\begin{fact}[Linear queries in \ptime~\cite{MeliouGMS11}]\label{linear query is easy}\label{coro:
    res linear query is easy} 
  For any linear sj-free CQ $q$, $\rsp(q)$ (and thus also $\res(q)$) are in \ptime.
\end{fact}

\begin{proof}
We give the proof for completeness and because we
will need an extension of the proof for a later result (\specificref{Lemma}{linear easy rspStar}).  

Let $q\datarule A_1(\vec{z_1})\land\cdots\land A_r(\vec{z_r})$ be a linear query, arranged in its
linear ordering.  We first show that $\res(q) \in \PTIME$. Let $D\models q$.  We construct a network
$N=N(q,D)$ as follows.
$N$ is an (r+1)-partite graph consisting of vertices $V = \set{s} \cup P_1 \cup P_2 \cup \cdots \cup
P_{r-1} \cup \set{t}$.  Each edge of $N$ has weight $1$ and corresponds
to exactly one tuple  $A_i(\vec a) \in D$. $P_i$ is the projection onto
$\var(A_i)\cap\var(A_{i+1})$ of $A_i^D \bowtie A_{i+1}^D$.  
The edge corresponding to $A_i(\vec a)$ is 
$(\pi_{\var(A_{i-1}) \cap \var(A_i)}(\vec a), \pi_{\var(A_{i}) \cap
  \var(A_{i+1})}(\vec a))$. However, $s$ is the starting point of all the $A_1$ edges, and $t$ is
the endpoint of all the $A_r$ edges (see \autoref{fig:NetworkFlow}).  

With this construction, a cut in $N(q,D)$ is exactly a contingency set for $(q,D)$ and thus a
min cut is exactly a minimum contingency set.  Thus we have reduced $\res(q)$ to network flow.

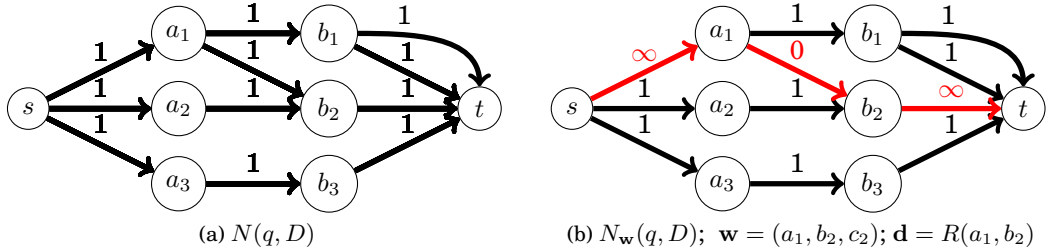
\begin{figure}
\begin{center}
\subfloat[\small{$N(q,D)$}]{
\begin{tikzpicture}[ scale=.25,
	every circle node/.style={fill=white, minimum size = 4mm, draw},
	]
\node (s)  at 	(0,0) [circle] {$s$};
\node (a1) at 	(8,4) [circle] {$a_1$};
\node (a2) at 	(8,0) [circle] {$a_2$};
\node (a3) at 	(8,-4) [circle] {$a_3$};
\node (b1) at 	(16,4) [circle] {$b_1$};
\node (b2) at 	(16,0) [circle] {$b_2$};
\node (b3) at 	(16,-4) [circle] {$b_3$};
\node (t)  at 	(24,0) [circle] {$t$};
\foreach \from/\to in {s/a1,s/a2,s/a3,a1/b1,a1/b2,a2/b2,a3/b3,b1/t,b2/t,b3/t}
\foreach \from/\to/\l in {s/a1/1,s/a2/1,s/a3/1,a1/b1/1,a1/b2/1,a2/b2/1,a3/b3/1,b1/t/1,b2/t/1,b3/t/1}
\draw[->,line width=2pt,color=black] (\from) -- (\to) node[pos=.5,above]{\l};
\draw[->,line width=2pt,color=black] (b1) to [out=0,in=90] (t);
\node at (20,5) {1};
\end{tikzpicture}
}
\hspace*{.15in}
\subfloat[\small{$N_{\vec w}(q,D)$; $\;\vec w = (a_1,b_2,c_2)$; $\vec d = R(a_1,b_2)$}]{
\begin{tikzpicture}[ scale=.25,
	every circle node/.style={fill=white, minimum size = 4mm, draw},
	]
\node (s)  at 	(0,0) [circle] {$s$};
\node (a1) at 	(8,4) [circle] {$a_1$};
\node (a2) at 	(8,0) [circle] {$a_2$};
\node (a3) at 	(8,-4) [circle] {$a_3$};
\node (b1) at 	(16,4) [circle] {$b_1$};
\node (b2) at 	(16,0) [circle] {$b_2$};
\node (b3) at 	(16,-4) [circle] {$b_3$};
\node (t)  at 	(24,0) [circle] {$t$};
\foreach \from/\to/\l in 
{s/a2/1,s/a3/1,a1/b1/1,a2/b2/1,a3/b3/1,b1/t/1,b3/t/1}
\draw[->,line width=2pt,color=black] (\from) -- (\to) node[pos=.5,above]{\l};
\foreach \from/\to/\l in
{s/a1/$\infty$,a1/b2/0,b2/t/$\infty$}
\draw[->,line width=2pt,color=red] (\from) -- (\to) node[pos=.5,above]{\l};
\draw[->,line width=2pt,color=black] (b1) to [out=0,in=90] (t);
\node at (20,5) {1};
\end{tikzpicture}
}
\end{center}
\caption{Network flow in the proof of \specificref{Fact}{linear query is easy} illustrated for query  
$q \datarule A(x),R(x,y),S(y,z)$ from \autoref{linear fig}
and database $D= \{ A,R,S \}$, 
where $A=\set{a_1,a_2,a_3}$,
$R=\set{(a_1,b_1),(a_1,b_2),(a_2,b_2),(a_3,b_3)}$,
  $S=\set{(b_1,c_1),(b_1,c_2),(b_2,c_2),(b_3,c_3)}$.  The drawing on the left is $N(q,D)$, the
  result of the reduction from
  $\res(q,D)$ to network flow.  The drawing on the right is $N_{\vec w}(q,D)$
  where we are computing the responsibility of $\vec d = R(a_1,b_2)$ and $\vec w =  (a_1,b_2,c_2)$.
\label{fig:NetworkFlow}}
\end{figure}

A similar but more complicated construction shows how to use network flow to compute the
responsibility of tuple $\vec d\in D$ 
for the linear query $q$.  We
construct the same network $N(q,D)$ but now we modify some of the edge weights.  We want to
compute the minimum size of a contingency set $\Gamma$ such that $D-\Gamma\models q$ but
$D-(\Gamma\cup{\vec d}) \not\models q$.  Consider all the witnesses $\vec w$ that $D\models q$ such that
$\vec w$ extends $\vec d$.  
For any contingency set $\Gamma$ for $\vec d$,  at least one such $\vec w$ must witness
$D-\Gamma\models q$.
Thus, $\Gamma$ must be disjoint from $\vec w$.  Observe that a contingency set
for $\vec d$ which is disjoint from $\vec w$ is a cut of $N(q,D)$ which removes $\vec d$ but leaves
the rest of $\vec w$.  The minimum weight of such a contingency set is exactly the min cut of $N_{\vec w}(q,D)$
which is formed from $N(q,D)$ by changing the weight of $\vec d$ to 0 (as it is removed at no cost) and changing the weights
of all the edges in $\vec w - \vec d$ 
to $\infty$: they cannot be removed.
Thus, the responsibility of $\vec d$ is the minimum over all witnesses $\vec w$ extending $\vec d$ of the
min cut of $N_{\vec w}(q,D)$.
We illustrate this construction for the query from \specificref{Example}{linear ex} in \autoref{fig:NetworkFlow}.

Thus we have shown that the complexity of computing $\res(q)$ is at most that of network flow.
On the other hand, $\rsp(q)$ may be computed by computing network flow of all the networks
$N_{\vec w}(q,D)$.  For each fixed $q$, there are at most ${O}(n^r)$ such  
$\vec w$.  Thus, for each $q$, $\rsp(q)\in\PTIME$.  
Note that for linear queries, the complexity of resilience is no more
than the complexity of network flow.  However, the complexity of resilience is in $\PTIME$ for each
fixed $q$, but we do not currently have a fixed upper bound on the size of the exponent.
\end{proof}

If all queries without a triad were linear, then this would complete the dichotomy theorem for resilience.
While this is not the case, we will show that \emph{any triad-free query can be transformed into a query
of equivalent complexity that is linear}.  

Recall that when studying resilience, we make atoms which are dominated, exogenous 
(\specificref{Prop.}{fact: domination does not change complexity}).  This is done, for example, to the rats
and brats queries, i.e., $\res(\rats) \equiv \res(\raxx)$ and $\res(\brats) \equiv \res(\brxxx)$
(see \specificref{Eq.}{eqn:raxx}).
Neither $\raxx$ nor $\brxxx$ is linear.  However they can be transformed to linear queries 
without changing their complexity via the following 
transformation from \cite{MeliouGMS11}:

\begin{definition}[Dissociation] \label{def: dissociation}
Let $A^\exSymb$ be an exogenous atom in a query $q$, and $v \in \var(q)$ a variable 
that does not occur in $A^\exSymb$. 
Let $q'$ be the same as $q$ except that we add $v$ to the arguments $A^\exSymb$. 
This transformation is called \emph{dissociation}.  
\end{definition}

\begin{example}[Dissociation]\label{dissociation linear ex}
The queries $\raxx$ and $\brxxx$ (\specificref{Eq.}{eqn:raxx}) have no triads but they are not
linear.  However, applying certain dissociations, we obtain
the following linear queries:
\begin{align*}
\raxx' & \datarule A(x),R^\exSymb(x,y,z),S(y,z),T^\exSymb(x,y,z)\\
\brxxx' & \datarule A(x),R^\exSymb(x,y,z),B(y)S^\exSymb(x,y,z),T^\exSymb(x,y,z)
\end{align*}
Note also that $\raxx'$ and $\brxxx'$ have duplicate atoms which we finally delete, without
affecting their complexity:
\begin{align*}
\raxx'' & \datarule A(x),R^\exSymb(x,y,z),S(y,z) \\
\brxxx'' & \datarule A(x),R^\exSymb(x,y,z),B(y)
\end{align*}
\end{example}

The key fact is that \emph{dissociation cannot decrease the complexity of resilience or
responsibility}.

\begin{lemma}[Dissociation increases complexity~\cite{MeliouGMS11}]\label{fact: dissociation do not decrease complexity}
If $q'$ is obtained from $q$ through dissociation, then $\res(q) \leq \res(q')$.
\end{lemma}

\begin{proof}
Let $R^\exSymb(\vec z)$ be the atom that has been changed to ${R^\exSymb}'(\vec z,v)$.
We  reduce $\res(q)$ to $\res(q')$ by mapping $(D,k)$ to $(D',k)$ where $D'$ is the same as $D$ with the
exception that we let ${R^\exSymb}' = \bigset{(\vec t,d)}{R^\exSymb(\vec t)\in D; d\in \mbox{dom}(D)}$. 
This transformation does not change the witness 
set nor the contingency sets, because, by the way we formed ${R^\exSymb}'$ from $R^\exSymb$, the conjunct ${R^\exSymb}'(\vec z,v)$
places the same restriction on $D'$ that $R^\exSymb(\vec z)$ places on $D$.
\end{proof}
	
The other direction does not hold, i.e,
dissociation may strictly increase the complexity of the resilience of a query\footnote{For example,
the query $\ell \datarule A(x),W_1^\exSymb(x,y),B(y),W_2^\exSymb(y,z),C(z)$ is linear, but by applying
dissociation we can transform it to $q_\Tri$.}.
It follows from \specificref{Lemma}{fact: dissociation do not decrease complexity} that if $q$ can be
dissociated to a linear query, then $\res(q)\in\PTIME$.  
In particular, the above dissociations of $\raxx$ and $\brxxx$ prove that 
$\res(\raxx)$ and $\res(\brxxx)$ are in $\PTIME$. Thus, since the
transformations from $\rats$ to $\raxx$ and $\brats$ to $\brxxx$ preserve the complexity of
resilience, we conclude that $\res(\rats)$ and $\res(\brats)$ are easy.
Later we will see that, for responsibility, $\rsp(\brats)\in\PTIME$ but $\rsp(\rats)$ is \NP-complete
(\specificref{Prop.}{responsibility of rats is hard}).

\begin{corollary}\label{res(rats) easy}
$\res(\rats)$ and $\res(\brats)$ are in $\PTIME$.
\end{corollary}

Later we will see that it is also true that dissociation does not decrease the complexity of responsibility, but
the proof is more subtle (\specificref{Lemma}{lem: dissociation rspStar}).

Now we are ready to show that the $\res(q)$ is easy if $q$ is triad-free.  We will show that for
every triad-free query, we can linearize the endogenous atoms and use 
some dissociations to make the exogenous atoms fit into the same order.

\begin{lemma}[Queries without triads are easy]\label{easy part dichotomy}
Let $q$ be an 
sj-free CQ that has no triad. Then $\res(q)$ is in \ptime.
\end{lemma}

\begin{proof}
Let $q$ be a triad-free query.
We prove by induction on the number of \emph{endogenous} atoms in $q$ that we can transform it into a linear query by using dissociations.   
Since dissociations cannot decrease complexity (\specificref{Lemma}{fact: dissociation do not decrease complexity})
and resilience is easy for linear queries (\specificref{Fact}{linear query is easy}), 
it follows that $\res(q)$ is in \ptime.

\emph{Base case}: $q$ has fewer than three endogenous atoms.
Consider $S_1, S_2$ the endogenous atoms of $q$. 
Using
dissociation, we add all the variables to all the exogenous atoms.
Thus all the exogenous atoms are identical and we can remove all but one, call it $E_1^\exSymb$.
The resulting query, $q'$, is linear with
ordering $S_1,E^\exSymb_1,S_2$.  Thus $\res(q)\in$ \ptime.

\emph{Inductive case}: assume true for triad-free queries with $n$ endogenous atoms.   Let $q_{n+1}$
be triad-free and have $n+1$ endogenous atoms. We now describe a way to linearize these atoms.
For each endogenous atom $S_i$, let $c_i$ be the cut of the hypergraph resulting from removing all
the variables of $S_i$, i.e., all the hyperedges that touch $S_i$.  These cuts are drawn as dotted vertical lines in \autoref{easyWalkFig}.  

Let $S_1$ and $S_2$ be two endogenous atoms and draw $S_2$ to the right of $S_1$. 
Now consider a third endogenous atom $S_3$.  Since $q_{n+1}$ is connected and has no triads, there
is a unique $i\in \set{1,2,3}$ such that the cut $c_i$ disconnects the two atoms in
$\set{S_1,S_2,S_3} - \set{S_i}$.

Thus we must place $S_i$ between the other two.  In other words, there is exactly one place that
$S_3$ can be added to the figure:  to the left of $S_1$ if $c_1$ separates $S_3$ from $S_2$; in
between $S_1$ and $S_2$ if $c_3$ separates $S_1$ from $S_2$; or to the right of $S_2$ if $c_2$
separates $S_1$ from $S_3$.

For example, let $S_1(x,y)$ and $S_2(y,z)$ be the first two endogenous atoms. Let the third be
$S_3(z,w)$ which shares a variable with $S_2$. Note that $c_3$ does not separate $S_1$ from $S_2$
and $c_1$ does not separate $S_2$ from $S_3$.  Since $q_{n+1}$ has no triad, it must be the case
that $c_2$ separates $S_1$ from $S_3$.  Thus, the order in this case must be $S_1, S_2, S_3$.

\begin{figure}
\begin{center}
\begin{tikzpicture}[ scale=.5]
\draw (2,2) circle [radius=.75] node {{\color{blue} $S_1$}};
\draw (2,-.75) node {{\color{black} $c_1$}};  
\draw (3,.65)  circle [radius=.75] node {{\color{dg} $E^\exSymb_1$}};  
\draw (4,-.75) node {{\color{black} $c_2$}};  
\draw (5,.65)  circle [radius=.75]  node {{\color{dg} $E^\exSymb_2$}};  
\draw (7,.75) node {{\color{dg} $\cdots$}};  
\draw (7,-.75) node {{\color{black} $\cdots$}};  
\draw (7,2) node {{\color{blue} $\cdots$}};  
\draw (9,.65)  circle [radius=.75]  node {{\color{dg} $E^\exSymb_{n-1}$}};  
\draw (10,-.75) node {{\color{black} $c_n$}};  
\draw (11,.65)   circle [radius=.75] node {{\color{dg} $E^\exSymb_{n}$}};  
\draw (12,-.75) node {{\color{black} $c_{n+1}$}};  
\draw (4,2) circle [radius=.75] node {{\color{blue} $S_2$}};
\draw (10,2) circle [radius=.75] node {{\color{blue} $S_{n}$}};
\draw (12,2) circle [radius=.75] node {{\color{blue} $S_{n+1}$}};
 \draw[line width=2pt,dotted] (2,2.75) -- (2,3.75); 
 \draw[line width=2pt,dotted] (2,1.25) -- (2,-.3); 
 \draw[line width=2pt,dotted] (4,2.75) -- (4,3.75); 
 \draw[line width=2pt,dotted] (4,1.25) -- (4,-.3); 
 \draw[line width=2pt,dotted] (10,2.75) -- (10,3.75); 
 \draw[line width=2pt,dotted] (10,1.25) -- (10,-.3); 
 \draw[line width=2pt,dotted] (12,2.75) -- (12,3.75); 
 \draw[line width=2pt,dotted] (12,1.25) -- (12,-.3); 

\end{tikzpicture}
\end{center}
\caption{A walk along the endogenous atoms in the proof of \specificref{Lemma}{easy part dichotomy}.  The cut $c_i$ results from removing all the variables
  (edges) from atom $S_i$.}
\label{easyWalkFig}
\end{figure}
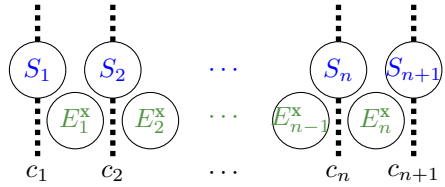

Now add the remaining endogenous atoms one at a time. 
Since $q_{n+1}$ has no triad, by the above observation, there is exactly
one place that each next endogenous atom may be placed.
Finally once all the endogenous atoms have been placed,  renumber them so left to right they are
$S_1$, $S_2$,
$\ldots$, $S_{n+1}$.  

Define the query $q_n$ to be the result of removing all the variables
in $\var(S_{n+1}) - \var(S_n)$ 
and removing all the atoms in which any of those removed variables occurred.
In \autoref{easyWalkFig}, this corresponds to removing everything to the right of $c_n$.

By our inductive hypothesis, 
there is a query $q_n'$ that is the result of doing
some dissociations to $q_n$,  and $q_n'$ is linear. Furthermore by our observation above, the ordering
of the endogenous atoms remains $S_1, S_2, \ldots, S_n$.

Now, we form $q'_{n+1}$ by first adding back to $q_n$ all the variables and atoms that we removed.
Note that we are thus adding back just one endogenous atom, $S_{n+1}$, together with zero or more
exogenous atoms, all of which contain some variables in $\var(S_{n+1}) - \var(S_n)$.  Finally, to
all these exogenous atoms that we have just added back (if any), add all the variables in $\var(S_n)
\cup \var(S_{n+1})$, together with any other variables occurring in any of these exogenous atoms.  
Thus all the newly re-added exogenous atoms are identical and we can combine them into one, call it,
$E^\exSymb_{n}$.   Note that $c_n$ still separates $E^\exSymb_{n}$ and $S_{n+1}$ from the rest
of the hypergraph.  

Thus, we have transformed $q_{n+1}$ to a linear query $q_{n+1}'$ such that 
$\res(q_{n+1}) \leq \res(q_{n+1}')$.  Thus $\res(q_{n+1}) \in\PTIME$ as desired.
\end{proof}

\subsection{Dichotomy of resilience}\label{sec:firstDichotomy}
Combining \specificref{Lemmas}{hard part dichotomy}  and
\ref{easy part dichotomy} leads to our
first dichotomy result on the complexity of resilience:

\begin{theorem}[Dichotomy of resilience]\label{resilience dichotomy thm}
Let $q$ be an sj-free CQ and let $q'$ be the result of 
making all dominated atoms exogenous.
If $q'$ has a triad, then $\res(q)$ is \NP-complete, otherwise
it is in \ptime.
\end{theorem}

Note that it is easy to tell whether $q$ has a triad.  Checking whether a given triple of atoms is a triad
consists of three reachability problems and -- is there a path from $S_i$ to $S_j$ not using any of
the edges in $\var(S_k)$ -- and is thus doable in linear time.

An exhaustive search of
all endogenous triples thus provides a \ptime\ algorithm:

\begin{corollary}\label{cor:ptimetriadcheck}
We can check in polynomial time in the size of the query $q$ whether $\res(q)$ is \NP-complete or \ptime.
\end{corollary}
\smallskip

%% file: full_4_Fds.tex
\section{Functional dependencies}\label{sec:FDs}

Functional dependencies (FDs), such as key constraints, restrict the set of allowable data instances. 
In this section, we characterize how these restrictions affect the complexity of
resilience.  We first show that FDs cannot increase the complexity of the resilience of a query
(\specificref{Prop.}{fd dont add complexity}).
Next we introduce a transformation of queries suggested by a given set of FDs call \emph{induced
  rewrites} (\specificref{Def.}{def: induced rewrite}).  We show that induced rewrites preserve the
complexity of resilience (\specificref{Lemma}{lm: induced rewrite preserves complexity}).

We call a query \emph{closed} if all possible induced rewrites have been applied (\specificref{Def.}{def: induced rewrite}).
We conjectured that induced rewrites capture the full power of FDs with respect to the complexity of
resilience, in other words, the complexity of the resilience of a closed query is unchanged if we remove
its FDs (\specificref{Conjecture}{conjecture}). 

We prove that the complexity of resilience for closed queries that have triads is \NP-complete (\specificref{Lemma}{hard part dichotomy with fd}).  On the other hand, even without its FDs, we know that a closed
query that has no triads has an easy resilience problem (\specificref{Lemma}{easy part dichotomy}).  We thus
conclude that in the presence of FDs,  the dichotomy -- still determined by the presence or absence
of triads, but now in the closure of the query -- remains in
force (\specificref{Lemma}{easy part dichotomy}).   It follows as a corollary that  \specificref{Conjecture}{conjecture} holds.

\subsection{FDs can only simplify resilience}
We write $\res(q; \setfd)$ to refer to the resilience problem for query $q$, restricted to databases
satisfying the set of FDs $\setfd$.  
Note that since we are always considering conjunctive queries, any particular FD either holds or
does not hold on the whole query, so it is not necessary to mention which atom the FD is applied to.

First we observe that FDs cannot make the resilience problem
harder:

\begin{proposition}[FDs do not increase complexity]\label{fd dont add complexity}
Let $q$ be an sj-free CQ and $\setfd$ a set of functional dependencies. Then $\res(q; \setfd) \leq \res(q)$.
\end{proposition}

\begin{proof}
The reduction is the identity function.  Note that
$\res(q; \setfd)$ is just the restriction of $\res(q)$ to
databases satisfying $\setfd$.  Thus, for all databases $D$ that satisfy $(q; \setfd)$:
$(D,k) \in \res(q;\setfd) \Leftrightarrow (D,k) \in \res(q)\;$.
\end{proof}

\begin{corollary}[Triad-free queries are still easy]\label{coro: ptime remains ptime}
If $q$ is an sj-free CQ that has no triad, and therefore $\res(q)$ is in \ptime, then
$\res(q; \setfd)$ is also in \ptime.
\end{corollary}

We next show that for some queries, FDs do in fact reduce the complexity of resilience. 
Recall that the tripod query, $q_\Tri $ is hard (\specificref{Prop.}{prop:tripodQuery}). 
However, $q_\Tri $ becomes polynomial when we add the
FD $\fd = x \to y$.

\begin{proposition}[FDs make $q_\Tri$ easy]\label{lem: tripod become easy}
\[\res(q_\Tri;\set{x\to y})\in \PTIME\; .\]
\end{proposition}

We will prove \specificref{Prop.}{lem: tripod become easy} along the way, as we learn about the effect of
FDs.   Recall that the tripod query $q_\Tri$
has the triad $\set{A,B,C}$.
Notice that the FD $x\to y$ ``disarms'' this triad because $A$ and $B$ are no longer independent. More
explicitly, once we know $x$, we also know $y$.  Thus $\res(q_\Tri;\set{x\to y}) \equiv \res(r)$
where $r\datarule A'(x,y),B(y),C(z), W^\exSymb(x,y,z)$
(\specificref{Lemma}{lm: induced rewrite preserves complexity}).  Furthermore, since $B$ dominates $A'$ in $r$,
$A'$ becomes exogenous: $r'\datarule {A'}^\exSymb(x,y),B(y),C(z), W^\exSymb(x,y,z)$.  Query $r'$ has no triad
and thus is easy.

\subsection{Induced rewrites preserve complexity}\label{sec:inducedRewrites}

We call the transformation $(q_\Tri;\set{x\to y})\rewrite(r;\set{x\to y})$ an \emph{induced
  rewrite}\footnote{Transformations of queries called \emph{rewrites} were defined in
  \cite{MeliouGMS11}. An induced rewrite is a rewrite that is induced by an FD.}.
Induced rewrites are key to understanding the effect of FDs on the complexity of
resilience.

\begin{definition}[induced rewrite: $\rewrite$, closed query]\label{def: induced rewrite}
Given a set of functional dependencies $\setfd$ and a query $q$, we write 
$(q; \setfd) \rewrite (q'; \setfd)$ to mean that $q'$ is the result of adding the 
dependent variable $u$ to some relation that contains all the determinant variables 
$\vec v$ for some $\vec v \to u \in \setfd$.
We use $\rewriteStar$ to indicate zero or more applications of $\rewrite$.
If $(q;\setfd) \rewriteStar (q^*;\setfd)$ and no more induced rewrites can be applied
to $(q^*;\setfd)$, then we call $(q^*;\setfd)$ a \emph{closed query} and we say that 
$(q^*;\setfd)$ is the \emph{closure} of $(q;\setfd)$.
\end{definition}

This paper began as an attempt to determine whether 
the dichotomy for responsibility of sj-free CQs \cite{MeliouGMS11} continues to hold in the presence
of FDs.  In studying the effect of FDs, we defined induced rewrites and proved 
that induced rewrites preserve the complexity of responsibility.  We conjectured that once we
have reached a closed query, all the effect of the FDs on the complexity of responsibility has
been exhausted and thus there is no further change if we delete all the FDs.  We were able to prove this
conjecture for unary FDs, i.e., those of the form $v\to u$ where $v$ is a single variable.

However we had great difficulty proving this conjecture for all FDs.  We studied the responsibility
problem more carefully and found that responsibility is quite delicate. In particular, we discovered
an error in Lemma 4.10 of \cite{MeliouGMS11}, namely that \specificref{Prop.}{fact: domination does not
  change complexity} (in the present paper) does not hold for responsibility.

We identified resilience as a better-behaved
notion than responsibility and we characterized the complexity of resilience via triads. Once we had
done that, we were able to use the notion of triads to prove our conjecture about closed queries
and thus prove the dichotomy theorem for resilience in the presence of arbitrary FDs.  We give that proof shortly.

With our improved insight from resilience, we went back and proved the dichotomy for
responsibility (\specificref{Theorem}{main rsp dichotomy thm}) and finally showed that it holds as well in the
presence of FDs (\specificref{Theorem}{last thm}).

We first show that induced rewrites preserve the complexity of resilience.

\begin{lemma}[Induced rewrites preserve complexity]\label{lm: induced rewrite preserves complexity}
Let $q$ be a query, $\setfd$ a set of functional dependencies, and $q'$ the result of an
induced rewrite, i.e., $(q;\setfd) \rewrite (q';\setfd)$. Then $\res(q'; \setfd)
\equiv \res(q; \setfd)$. 
\end{lemma}

\begin{proof}
Let the change from $q$ to $q'$ be the transformation of the atom $B$ to the new atom $B'$
caused by adding variable $u$ to $B$ where $(\vec{v}
\to u) \in \setfd$ and $\vec v \subseteq \var(B)$.
\begin{enumerate}[itemsep=2pt, topsep=2pt]
\item[(a)] $\res(q'; \setfd) \leq \res(q; \setfd)$:
Suppose we are given $(D',k)$ where $D'$ satisfies $\setfd$.  
Let $D$ be the result of projecting out the $u$ entry from $B'$.
Note that $D$ still satisfies $\setfd$.  Furthermore, 
the set of witnesses that  $D\models q$ is identical to the set of witnesses that $D'\models q'$ 
and the sizes of all minimum contingency sets are unchanged. This is because the effect of the tuple
$B(\vec t)$ in a contingency set in $D$ is identical to the effect of the tuple $B'(\vec t')$ in the
corresponding contingency set in $D'$, where $\vec t'$ is the result of adding to $\vec t$ the
unique $u$-attribute which is determined by the $\vec v$-attributes of $\vec t$.
Thus the map $(D',k) \mapsto
(D,k)$ is a reduction of $\res(q'; \setfd)$ to $\res(q; \setfd)$.
\item[(b)] $\res(q; \setfd) \leq \res(q'; \setfd)$.
We are given $(D,k)$ where $D$ satisfies $\setfd$.
Let $B'$ be the set of tuples resulting from adding to each tuple $\vec t$ from $B$, the uniquely
determined $u$-attribute, $c$.  In symbols, $B'=$
\[
\bigset{(\vec t,c)}{B(\vec t) \in D \sland \exists \vec s \in D \,(\pi_{\vec v}(\vec s) = \pi_{\vec
    v}(\vec t) \land c =
  \pi_u(\vec s))}
\]
For the same reason as above, the witnesses of $q'$ in $D'$ are the same as the witnesses of $q$ in
$D$ and the sizes of all minimum contingency sets are unchanged.  Thus the map $(D,k) \mapsto
(D',k)$ is a reduction of $\res(q; \setfd)$ to $\res(q'; \setfd)$.
\qed
\end{enumerate}
\end{proof}

It follows immediately that applying any set of induced rewrites preserves the complexity of resilience:

\begin{corollary}\label{coro: closed query preserves complexity}
If $(q;\setfd) \rewriteStar (q';\setfd)$,  then $\res(q'; \setfd)
\equiv \res(q; \setfd)$.
\end{corollary}

\subsection{For closed queries, FDs are superfluous}\label{sec:closedQueries}

Recall that our current goal is to determine whether the dichotomy of the complexity of
resilience remains true in the presence of FDs.  The following is a natural conjecture which would
given an affirmative answer to this question.

\begin{conjecture}[Induced rewrites suffice]\label{conjecture}
Let $(q^*;\setfd)$ be a closed query, i.e., it is closed under induced rewrites.
Then $ \res(q^*; \setfd)\equiv \res(q^*)$. 
\end{conjecture}

It is fairly easy to see that Conjecture \autoref{conjecture} holds when all the FDs in $\setfd$ are unary,
i.e., of the form $v\to u$, with $v$ a single variable.  However we were stumped about how to prove this
for general FDs.  This lead to our more careful analysis of the complexity of responsibility, our
definition of resilience, and our characterization of the complexity of resilience via triads
(\autoref{resilience dichotomy thm}). Now we will use that analysis to prove that the complexity of
a closed query is \NP-complete if it contains a triad, and in $\PTIME$ otherwise.
Thus Conjecture
\autoref{conjecture} is true and the dichotomy for the complexity of resilience remains true in the
presence of FDs.

\begin{lemma}[Closed queries with triads are hard]\label{hard part dichotomy with fd}
Let $(q^*;\Phi)$ be a closed sj-free CQ all of whose dominated atoms are exogenous.
If $q^*$ has a triad, then $\res(q^*; \setfd)$ is \NP-complete.
\end{lemma}

\begin{proof}
Let $(q^*;\Phi)$ be as in the statement of the lemma.
Recall that we proved in \specificref{Lemma}{hard part dichotomy} that $\res(q_\triangle) \leq \res(q^*)$ and
thus $\res(q^*)$ is \NP-complete.  Let $f$ be the reduction we produced from $\res(q_\triangle) $ to 
$\res(q^*)$.  We will now show that if $f(D,k) = (D',k')$ then $D'\models \Phi$.  
It will then
follow that $f$ is a reduction from $\res(q_\triangle) $ to $\res(q^*;\Phi)$.  Thus
$\res(q^*;\Phi)$ is \NP-complete as claimed.

To see why $D'\models \Phi$, we will recall the definition of the reduction in the proof of
\specificref{Lemma}{hard part dichotomy}.  But first, we will examine how $q_\triangle$ (\specificref{Example}{four main
  queries ex}) itself is affected by FDs.

In particular, let $\Phi_0$ be any set of FDs for which $(q_\triangle, \Phi_0)$ is closed under
induced rewrites.  Notice that since
$q_\triangle$ is closed, there can be no nontrivial unary FDs such as $x\rightarrow y$, (otherwise,
$T(z,x)$ would have been replaced by $T'(z,x,y)$) 
nor any
nontrivial binary FDs such as $xy\to z$ (otherwise $R(x,y)$ would have been replaced by $R'(x,y,z)$).
In fact, $\Phi_0$ has no nontrivial FDs, i.e.,
$\Phi_0=\emptyset$.  

Now recall the reduction from $\res(q_\triangle)$ to $\res(q^*)$ in the proof of 
\specificref{Lemma}{hard part dichotomy}. What that proof did was to embed $q_\triangle$ into $q^*$.  Using the triad of $q^*$,
${\cal T}=\set{S_0,S_1,S_2}$, we partitioned the variables of $q^*$ into 7 sets,
and for each assignment of $x,y,z$ to values $a,b,c\in\dom(D)$, we made assignments
according to that partition (see \autoref{variable partition eq}).

The net effect, is that just as for $q_\triangle$, since $(q;\Phi)$ is closed, it must be the case
that $D' \models \Phi$.  In particular, suppose that $\Phi$ contains the FD, $\vec u\to v$.
First suppose that $\vec u$ is
contained in one of the 7 sets of the partition (see \autoref{variable partition eq}). Then, since $(q^*;\Phi)$
is closed, $v$ must be in the same 
set and thus it has exactly the same value as each of the variables in $\vec u$.  
If $\vec u$ has a variable from $V_3$ ($\var(q) - (\var(S_0)\cup \var(S_1)\cup \var(S_2))$)
then its value is $\angle{abc}$ so it determines all other variables. Similarly, if $\vec u$ has
variables from two of $V_0, V_1, V_2$ then it again determines all three values.  
Suppose $\vec u$ does not determine all three values, e.g., say it does not determine $c$.  Then,
looking at \autoref{variable partition eq}, we see that all the variables of $\vec u$ are
from $V_0, V_4$ or $V_5$, i.e., they are all from $\var(S_0)$.  But then since $(q^*;\Phi)$
is closed, $v$ must be in $\var(S_0)$ as well, and thus it is determined by $a$ and $b$.

Thus, we have shown that the reduction $f$ is also a reduction from $\res(q_\triangle)$ to
$\res(q^*,\Phi)$ and thus the latter problem is \NP-complete.
\end{proof}

\subsection{Dichotomy of resilience with FDs}\label{sec:secondDichotomy}

Recall that FDs cannot increase the complexity of resilience and thus if $q$ has no triad, then
$\res(q;\Phi)\in\PTIME$ (\specificref{Cor.}{coro: ptime remains ptime}).  Thus, we have succeeded in proving
the dichotomy for resilience in the presence of FDs:

\begin{theorem}[FD Dichotomy]\label{fd dichotomy thm}
Let $(q;\Phi)$ be an sj-free 
CQ with functional dependencies.  Let $(q^*,\Phi)$ be its
closure under induced rewrites, and such that all dominated atoms of $q^*$ are exogenous. 
If $q^*$ has a triad then $\res(q;\Phi)$ is \NP-complete.  Otherwise,
$\res(q;\Phi)\in\PTIME$.
\end{theorem}

Note that we have thus also proved \specificref{Conjecture}{conjecture}:

\begin{corollary}[Induced rewrites suffice]\label{co:inducedRewritesSuffice}
Let $(q;\Phi)$ be an sj-free CQ with functional dependencies, and
let $q^*$ be the closure of $q$ under induced rewrites. Then, $\res(q;\setfd) \equiv \res(q^*;\setfd) \equiv \res(q^*)$.
\end{corollary}

%% file: full_5_Responsibility.tex
\section{Complexity of Responsibility}\label{responsibility sec}

We now develop and prove the analogous characterizations of the complexity 
of responsibility. As we will see, responsibility is a bit more delicate than resilience, but in the
end the final theorems are similar.

We first concentrate on the difference between resilience and responsibility.
Recall the queries $\rats$ and $\raxx$ (\specificref{Example}{four main queries ex} and  \specificref{Eq.}{eqn:raxx}).
We saw earlier that $\res(\rats)$ is in \PTIME (\specificref{Cor.}{res(rats) easy}).  The reason is
that atom $A$ dominates $R$ and $T$ and thus the complexity of $\res(\rats)$ is unchanged when we
make $R$ and $T$ exogenous (\specificref{Prop.}{fact: domination does not change complexity}), i.e., $\res(\rats)
\equiv \res(\raxx)$.  Obviously $\raxx$ is triad-free.  
Thus, by \autoref{resilience dichotomy thm}, $\res(\raxx)$ and $\res(\rats)$ are in
$\PTIME$.  We now show, however, that $\rsp(\rats)$ is \NP-complete.

\begin{proposition}[$\rats$ is hard for $\rsp$]\label{responsibility of rats is hard}
$\rsp(\rats)$ is \NP-complete.
\end{proposition}

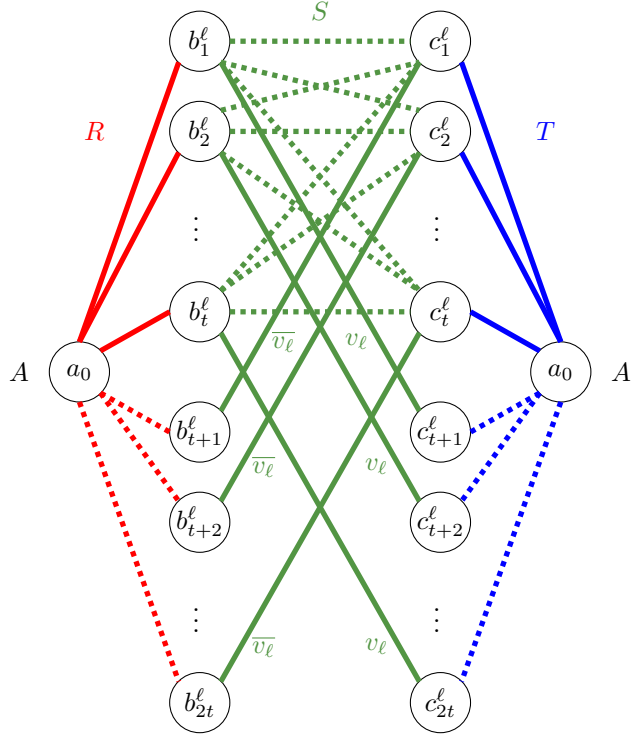
\begin{figure}
\begin{center}
\begin{tikzpicture}[ scale=.4]
\draw (-7.5,8) node[color=red] { $R$};	
\draw (0,12) node[color=dg] { $S$};	
\draw (7.5,8) node[color=blue] { $T$};	
\draw (-10,0) node[color=black] { $A$};	
\draw (10,0) node[color=black] { $A$};	
\draw (-8,0) circle [radius=1.0] node { $a_0$};
\draw (8,0) circle [radius=1.0] node { $a_0$};
\draw (-4,11) circle [radius=1.0] node { $b^{\ell}_{1}$};
\draw (-4,8) circle [radius=1.0] node { $b^{\ell}_{2}$};
\draw (-4,5) node {{\color{black} $\vdots$ }};
\draw (-4,-8) node {{\color{black} $\vdots$ }};
\draw (4,-8) node {{\color{black} $\vdots$ }};
\draw (4,5) node {{\color{black} $\vdots$ }};
\draw (-4,2) circle [radius=1.0] node { $b^{\ell}_{t}$};
\draw (-4,-2) circle [radius=1.0] node { $b^{\ell}_{t+1}$};
\draw (-4,-5) circle [radius=1.0] node { $b^{\ell}_{t+2}$};
\draw (-4,-11) circle [radius=1.0] node { $b^{\ell}_{2t}$};
\draw (4,11) circle [radius=1.0] node { $c^{\ell}_{1}$};
\draw (4,8) circle [radius=1.0] node { $c^{\ell}_{2}$};
\draw (4,2) circle [radius=1.0] node { $c^{\ell}_{t}$};
\draw (4,-2) circle [radius=1.0] node { $c^{\ell}_{t+1}$};
\draw (4,-5) circle [radius=1.0] node { $c^{\ell}_{t+2}$};
\draw (4,-11) circle [radius=1.0] node { $c^{\ell}_{2t}$};
 \draw[line width=2pt,color=red] (-8,1) -- (-4.7,10.3)  ; 
 \draw[line width=2pt,color=red] (-8,1) -- (-4.7,7.3) ; 
 \draw[line width=2pt,color=red] (-7.3,.7) -- (-5,2)  ; 
 \draw[line width=2pt,color=red,dotted] (-7.3,-.7) -- (-5,-2); 
 \draw[line width=2pt,color=red,dotted] (-7.3,-.7) -- (-4.7,-4.3); 
 \draw[line width=2pt,color=blue,dotted] (7.3,-.7) -- (4.7,-4.3); 
 \draw[line width=2pt,color=red,dotted] (-8,-1) -- (-4.7,-10.3); 
 \draw[line width=2pt,color=blue] (8,1) -- (4.7,10.3)  ; 
 \draw[line width=2pt,color=blue] (8,1) -- (4.7,7.3) ; 
 \draw[line width=2pt,color=blue] (7.3,.7) -- (5,2)  ; 
 \draw[line width=2pt,color=blue,dotted] (7.3,-.7) -- (5,-2); 
 \draw[line width=2pt,color=blue,dotted] (8,-1) -- (4.7,-10.3) ; 
 \draw[line width=2pt,dotted,color=dg] (-3,11) -- (3,11) ; 
 \draw[line width=2pt,dotted,color=dg] (-3.3,10.3) -- (3.3,8.7) ; 
 \draw[line width=2pt,dotted,color=dg] (-3.3,8.7) -- (3.3,10.3) ; 
 \draw[line width=2pt,dotted,color=dg] (-3,8) -- (3,8)  ; 
 \draw[line width=2pt,dotted,color=dg] (-3,2) -- (3,2);  
 \draw[line width=2pt,dotted,color=dg] (-3.3,7.3) -- (3.3,2.7)  ; 
 \draw[line width=2pt,dotted,color=dg] (-3.3,2.7) -- (3.3,7.3)  ; 
 \draw[line width=2pt,dotted,color=dg] (-3.3,2.7) -- (3.3,10.3)  ; 
 \draw[line width=2pt,dotted,color=dg] (-3.3,10.3) -- (3.3,2.7)  ; 
 \draw[line width=2pt,color=dg] (-3.3,10.3) -- (3.3,-1.3) node[pos=.8,left]{$v_\ell$}; 
 \draw[line width=2pt,color=dg] (-3.3,7.3) -- (3.3,-4.3) node[pos=.9,left]{$v_\ell$}; 
 \draw[line width=2pt,color=dg] (-3.3,1.3) -- (3.3,-10.3)   node[pos=.9,left]{$v_\ell$}; 
 \draw[line width=2pt,color=dg] (-3.3,-1.3) -- (3.3,10.3) node[pos=.2,right]{$\ov{v_\ell}$}; 
 \draw[line width=2pt,color=dg] (-3.3,-4.3) -- (3.3,7.3) node[pos=.1,right]{$\ov{v_\ell}$}; 
 \draw[line width=2pt,color=dg] (-3.3,-10.3) -- (3.3,1.3)   node[pos=.1,right]{$\ov{v_\ell}$}; 
\end{tikzpicture}
\end{center}
\caption{The $\rats$ variable gadget $G_\ell$ for variable $v_\ell$. Red, green, and blue lines correspond to tuples from $R$, $S$, and $T$, respectively.
Dotted lines 
will never need to be chosen in minimum contingency sets of $f(\psi)$.}
\label{rats gadget fig}
\end{figure}

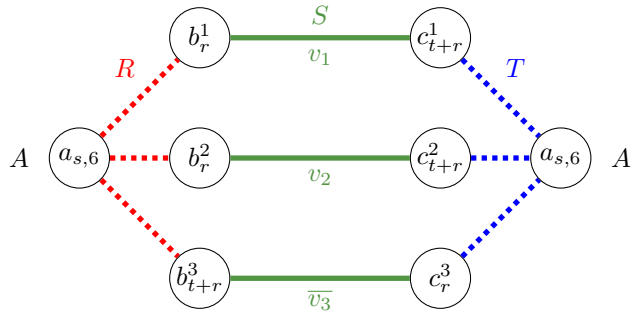
\begin{figure}
\begin{center}
\begin{tikzpicture}[ scale=.4,
			every circle node/.style={fill=white, minimum size=8mm, inner sep=0,draw}]
\draw (-10,0) node[color=black] { $A$};	
\draw (10,0) node[color=black] { $A$};	
\draw (0,4.75) node[color=dg] { $S$};	
\draw (-6.5,3) node[color=red] { $R$};
\draw (6.5,3) node[color=blue] { $T$};	
\node (L) at (-8,0)   [circle] { $a_{s,6}$};
\node (R) at (8,0)   [circle] { $a_{s,6}$};
\node (b1) at  (-4,4) [circle]  { $b^1_{r}$};
\node (b2) at  (-4,0) [circle]  { $b^2_{r}$};
\node (b3) at  (-4,-4) [circle]  { $b^3_{t+ r}$};
\node (c1) at  (4,4) [circle]  { $c^1_{t+r}$};
\node (c2) at  (4,0) [circle]  { $c^2_{t+r}$};
\node (c3) at  (4,-4) [circle]  { $c^3_{r}$};
\foreach \from/\to in {L/b1,L/b2,L/b3}
\draw[line width=2pt,color=red,dotted] (\from) -- (\to);
\foreach \from/\to/\l in {b1/c1/$v_1$,b2/c2/$v_2$,b3/c3/$\ov{v_3}$}
\draw[line width=2pt,color=dg] (\from) -- (\to) node[pos=.5,below]{\l}; 
\foreach \from/\to in {c1/R,c2/R,c3/R}
\draw[line width=2pt,color=blue,dotted] (\from) -- (\to);
\end{tikzpicture}
\end{center}

\caption{The $\rats$ clause gadget corresponding to clause $C_s = v_1 \lor \ov{v_2} \lor v_3$ and truth assignment $\alpha_6=\set{\angle{v_1,1},\angle{v_2,1},\angle{v_3,0}}$.
$A(a_{s,6})$ must be in the minimum contingency set unless the chosen truth assignment is
$\alpha_6$.}
\label{rats clause fig}
\end{figure}

\begin{proof}
We reduce 3SAT to $\rsp(q_{\textrm{rats}})$.
Let $\psi$ be a 3-CNF formula with variables $v_1, \ldots, v_n$ and clauses $C_1, \ldots, C_m$.
The reduction will map $\psi$ to $f(\psi) = (D,\vec{s_0},k)$ with $\vec{s_0} = S(b_0,c_0)$, where we will construct
$D=(A,R,S,T)$ to have a contingency set 
for $\vec{s_0}$ 
of size $k$ iff $\psi\in$ 3SAT (we explain the choice of value $k$ later in the proof). We let $a_0$ be the unique element of the
domain of $D$ that joins with $\vec{s_0}$.  

In $\rats$, $A$ dominates $R$, but when we are building a contingency set $\Gamma$ for $\vec{s_0}$,
we may require some tuples of the form $R(a_0,b)$.  Note that these cannot be 
replaced by the tuple $A(a_0)$, because that would remove the only witness $(a_0, b_0, c_0)$
that contains our tuple $\vec s_0$.
This  explains why $\res(\rats)\in\PTIME$ while $\rsp(\rats)$ is \NP-complete, 
and it is the key idea behind the reduction we now produce.

For each variable $v_\ell$ occurring in $\psi$, we build the gadget $G_\ell$ as follows: $G_\ell$
consists of $2t$ $b^{\ell}_{j}$ values for $y$ and $2t$ $c^{\ell}_{j}$ values for $z$ ($1\leq j\leq 2t$)
where
$t$ is a constant to be specified later.  We include the $2t$ pairs $R(a_0,b^{\ell}_{j})$ and the $2t$ pairs
$T(c^{\ell}_{j},a_0)$, $1 \leq j \leq 2t$. (See \autoref{rats gadget fig} where these pairs are drawn as edges 
from $a_0$ to each $b^{\ell}_{j}$ and  from each $c^{\ell}_{ j}$ to $a_0$, respectively. 
Notice that the value $a_0$ is shown twice for better illustration.)

Next, we include all the pairs $S(b^{\ell}_{j},c^{\ell}_{j'})$, $1\leq j,j' \leq t$.  These are
drawn in \autoref{rats gadget fig} as a complete bipartite graph between the vertex sets $\set{b^{\ell}_{1},
  \ldots,b^{\ell}_{t}}$ and $\set{c^{\ell}_{1}, \ldots,c^{\ell}_{t}}$.

Finally we add two matchings of size $t$ which we name the ``$v_\ell$ matching'' and the ``$\ov{v_\ell}$ matching,'' respectively:
\begin{align*}
	v_\ell \; \textrm{matching}:\quad &  
		S(b^{\ell}_{ 1},c^{\ell}_{t+1}), \ldots, S(b^{\ell}_{  t},c^{\ell}_{2t})\\
	\ov{v_\ell}\; \textrm{matching}:\quad &  
		S(b^{\ell}_{t+1},c^{\ell}_{1}), \ldots, S(b^{\ell}_{2t},c^{\ell}_{t})
\end{align*}
Notice that in \autoref{rats gadget fig}, 
the $v_\ell$ matchings are connecting the upper left corner with the lower right corner, whereas the $\ov{v_\ell}$ matchings are connecting the other two corners.

Any minimum contingency set must remove all of the witnesses from $G_\ell$.  
Such a minimum contingency set must remove either all the pairs $R(a_0,b^\ell_{1}), \ldots
R(a_0,b^{\ell}_{t})$ or all the pairs $T(c^{\ell}_{1},a_0), \ldots T(c^{\ell}_{t},a_0)$, i.e., one side or
the other of the complete bipartite graph.  After this, $t$ witnesses remain, either involving the 
$v_\ell$ matching (if the $T(c^{\ell}_{i},a_0)$'s were chosen), or otherwise the $\ov{v_\ell}$
matching.  Only the $S$-tuples will be useful for the clause gadgets, so the optimal choice will be
to choose the $t$ $S$-tuples marked $v_\ell$ or the $t$ $S$-tuples marked $\ov{v_\ell}$.  Any
optimal minimal contingency set thus corresponds to a truth assignment to the boolean variables
$v_1, \ldots, v_n$.

So far, we have described the gadgets $G_1, \ldots G_n$ and shown that any minimum contingency set
for this part of $D$ corresponds to a truth assignment for the variables $v_1, \ldots, v_n$.  We
next introduce the clause gadgets and choose the value $k$, so that contingency sets for $D$ of size $k$
will correspond exactly to truth assignments that satisfy all of the clauses of $\psi$.

We now describe the clause gadgets.  Suppose, for example, that
$C_s = v_1 \lor \ov{v_2} \lor v_3$ with $s\in[m]$.  
Then 7 of the eight possible truth assignments to $v_1,v_2,v_3$
satisfy $C_s$, i.e., all but the assignment $\alpha_2$ (010 in binary).  For each of these 7 good assignments:
$\alpha_i$, $i \in \{0, \ldots 7\} - \{2\}$, 
we add an element $a_{s,i}$ to $A$ and we add the tuples to $R$ and $T$ 
so that $a_{s,i}$ participates in three witnesses, each of which shares an $S$ tuple with a witness from each of the three variable gadgets that agree
with assignment $\alpha_i$. 
For example, assignment $\alpha_6$ (110 in binary) makes $v_1,v_2$ true  and $v_3$ false, so $a_{s,6}$ joins with
$S(b^1_{r(s,6)},c^1_{t+r(s,6)})$, 
$S(b^2_{r(s,6)},c^2_{t+r(s,6)})$, and 
$S(b^3_{t+ r(s,6)},c^3_{r(s,6)})$.  
Here $r(s,i)$ is a function that chooses a unique element of the matching $v_j$ or $\ov{v_j}$
appropriate to assignment $\alpha_i$ of clause $s$ (see \autoref{rats clause fig}).

The key property of the $C_s$ gadget is that, if the chosen truth assignment satisfies $C_s$, then we
do not need to worry about the
$a_{s,i}$ corresponding to the chosen assignment, and may choose only 6  $a_{s,i}$'s from $A$ for the contingency set.  
However, if the chosen assignment does not
satisfy $C_s$, then all 7 of the $a_{s_i}$'s must be chosen!

We can let $t=8m$ and $k=(2t)n + 6m$ = $(16n + 6)m$.
Our construction insures that $(D,\vec{s_0},k)\in
\rsp(\rats)$ iff $\psi\in$ 3SAT.
\end{proof}

Notice that in the proof of \specificref{Prop.}{responsibility of rats is hard} we showed that is hard to
compute the responsibility for a tuple from $S$ in  $\rsp(\rats)$.  The complexity of computing the
responsibility of a tuple can depend on which relation the tuple is chosen from.  In the case of
$\rats$, responsibility is hard for tuples from all relations except for $A$.

The proof of \specificref{Prop.}{responsibility of rats is hard} shows that domination does not work the same
way for responsibility as it does for resilience.  In particular, 
the analogy of \specificref{Prop.}{fact:
  domination does not change complexity} (Domination for Resilience) does not hold for responsibility.  

We next show that a modified version of domination still works for responsibility.  Recall the queries
$\brats$ (\specificref{Example}{four main queries ex}) and define the query $\brx$ as follows:

\myequation{eqn:brx}{\brx \datarule A(x),R^\exSymb(x,y),B(y),S(y,z),T(z,x)\; .}

Notice that
$\var(A) \subset \var(R)$ and $\var(B) \subset \var(R)$ and that also $\var(R) 
\subseteq \var(A) \cup \var(B)$.

\begin{proposition}[$\rsp(\brats)$]\label{full domination in brats}
The complexity of responsibility for $\brats$ is unchanged if we make $R$ exogenous, i.e., 
\[\rsp(\brats) \equiv \rsp(\brx)\; .\]
\end{proposition}

\begin{proof}
Let $D\models \brats$ and let $\vec t$ be a tuple that participates in a witness that $D\models\brats$.
We will show that there is a minimum contingency set $\Gamma'$ 
for $\vec t$ that contains no tuples from $R$.
Let $\Gamma$ be a minimum contingency set for $\vec t$ that contains as few tuples from $R$ as
possible. Suppose that $R(a_1,b_1)\in \Gamma$.
Let $\vec{j}$ be a witness that $(D-\Gamma)\models\brats$ and 
let $a_0,b_0,c_0$ be the projection of $\vec{j}$ onto components $x,y,z$, respectively.  
Thus, $A(a_0),R(a_0,b_0)$ and $B(b_0)$ are all in $D-\Gamma$.
In particular, $R(a_1,b_1)\ne R(a_0,b_0)$.
Let $\Gamma'$ be the result of replacing $R(a_1,b_1)$ by $A(a_1)$ if
$a_1\ne a_0$, and by $B(b_1)$ otherwise, in which case $b_1 \ne b_0$.
Thus $\Gamma'$ is still a minimum contingency set for $\vec t$ and it contains fewer tuples from $R$,
contradicting the fact that $\Gamma$ had the fewest possible such tuples.
Thus, tuples from $R$ are never needed in any minimum contingency set for $\vec t$.
Thus, as claimed, the complexity of $\rsp(\brats)$ is unchanged when we make $R$ exogenous.
\end{proof}

We are now ready to formalize \emph{full domination}, the version of domination that works for responsibility the way
that ordinary domination works for resilience.  Our first example is that in the query $\brats$, the relation $R$ is fully
dominated  because every variable in $\var(R)$ is ``covered'' by some other endogenous
relation (\specificref{Prop.}{full domination in brats}).\footnote{
Contrast this with the definition of domination (\specificref{Definition}{domination}) which only
requires that some subset of the variables is covered by another relation.}
Here are three more examples, $s_1,s_2,s_3$ where $R$ is fully
dominated and one, $n_4$, where it is not.  

\myequation{eqn:fullDominationExamples}{
\begin{array}{rcl}
s_1 &\datarule& A(x),R(x,y,w),B(y),S(y,z),T(z,x)\\
s_2 &\datarule& A(x),R(x,y,w),Q^\exSymb(w),B(y),S(y,z),T(z,x)\\
s_3 &\datarule& A(x),R(x,y,w),Q^\exSymb(w,x),B(y),S(y,z),T(z,x)\\
n_4 &\datarule& A(x),R(x,y,w),Q^\exSymb(w,z),B(y),S(y,z),T(z,x)\\
\end{array}}

In a query $q$,  call a variable $w\in\var(R)$ \emph{solitary} if it cannot reach another endogenous atom
without following one of the edges in $\var(R) - \set{w}$.  Note that in each of $s_1,s_2,s_3$, the
variable $w$ is solitary, but $w$ is not solitary in $n_4$.

\begin{definition}[Full domination]\label{full domination}
Let $F$ be an atom of query $q$.  $F$ is
\emph{fully dominated} iff for all non-solitary variables $y\in \var(F)$
there is another atom $A$ such that $y\in\var(A)\subset\var(F)$.
\end{definition}

Observe that relation $R$ is fully dominated in $\brats$, as well as in $s_1,s_2,s_3$, but not in $n_4$ (\specificref{Eq.}{eqn:fullDominationExamples}).
On the other hand, $R$ is not fully dominated
in $\rats$ because $y$ is connected to $S(y,z)$ and thus not
solitary and not covered by any smaller atom.

We now show that fully dominated atoms may be made exogenous.

\begin{lemma}[Full domination]\label{full domination vs exogenous}
Let $F$ be a fully dominated atom in an sj-free CQ  $q$.  Let $q'$ be the modified query in which $F$
is made exogenous.  Then $\rsp(q) \equiv \rsp(q')$.
\end{lemma}

{\sloppy

\begin{proof}
We have to show that $\rsp(q) \leq \rsp(q')$ and $\rsp(q') \leq \rsp(q)$. Suppose we are given
$(D,S(\vec t))$ and we are interested in the responsibility of tuple $S(\vec t)$.
There are
two cases. In each case, we will show how, given one of $k,k'$, to produce the other, such that:
\myequation{full domination eq}{(D,\vec t,k)\in\rsp(q) \qLra (D',\vec{t},k') \in \rsp(q')}

\emph{Case 1}: $F\ne S$:   We show that as in the proof
of \specificref{Prop.}{full domination in brats}, there is no need to include any tuples from $F$ in a minimum
contingency set $\Gamma$ for $q$ in $D$.  As in that proof, we let 
$\vec{j}$ be a witness for  $(D-\Gamma)\models q$ and suppose that $F(\vec f) \in \Gamma$.
Thus, $\vec j$ and $\vec f$ must disagree on the assignment of at least one variable.

\emph{(a)}: 
Suppose they differ on some non-solitary variable $y$ of $F$.  Let $A$ be the
atom that covers $y$ and we can replace $F(\vec f)$ by the tuple $\pi_{\var(A)}(\vec f)$ of $A$.
Thus, the sizes of the minimum contingency sets on the two sides are identical and 
letting $k = k'$ and $D=D'$, \specificref{Eq.}{full domination eq} holds. 

\emph{(b)}: 
Suppose on the contrary that $\vec j$ and $\vec f$ agree on all the non-solitary variables of
$F$. Note that since $S$ is endogenous, 
no non-solitary variable of $F$ can occur in $S$\footnote{We are allowing the computation of the
  responsibility of tuples from exogenous relations just to make the proofs simpler.  Notice
  that we never change the relation $S$ whose tuples we are computing the responsibility of.  Thus,
  if we must make $S$ exogenous, we do so as the last fully-dominated atom we make exogenous.}.
Thus, the only place that  $\vec j$ and $\vec f$ 
disagree is on non-solitary variables of $F$ which do not occur in $S$.  Let $F(\vec{f_0})$ be the
tuple of $F$ that agrees with $\vec j$.  Then $\vec f$ and $\vec{f_0}$ agree on all variables except
for solitary variables of $F$.  Thus, since removing $S(\vec t)$ from $D-(\Gamma-\set{F(\vec f)})$ removes all
witnesses of $D\models q$ that extend $\vec{f_0}$, it must also remove all witnesses that extend
$\vec f$, i.e., $\vec f$ is not useful so it does not occur in $\Gamma$.

\emph{Case 2}: $F = S$: 
In this case, some tuples of $F$ may need to be in
$\Gamma$. Let $I$ be the solitary variables of $F$ and let $W = \bigset{\vec f\in F}{\vec f \mbox{
    useful
; } \vec f \ne \vec   t \wedge \pi_{\ov     I}(\vec f)=\pi_{\ov I}(\vec t)}$. 
These are the tuples of $F$ which agree with $\vec t$ on all
but the solitary variables of $F$.  
$W$ must be contained in every contingency set for $(D,\vec t)$.  Thus, we
let $k= k' + \abs{W}$ and $F' = F - W$.  \specificref{Eq.}{full domination eq} holds.
(The point of $\vec f$ being useful in the definition of $W$ is that solitary variables may occur in some exogenous relations which
could already exclude certain values, and thus tuples with those values are not useful so they do
not need to be in the contingency set.)
\end{proof}

}

\subsection{Triads and hardness} 

Now that we have established that full domination works for responsibility, we proceed to prove a
complexity dichotomy for responsibility.  

When studying responsibility, we will insist from now on that every fully dominated atom is
exogenous. For example, $\rats$ has no fully dominated atoms, so it
is already in its normal form and it has a triad, $\set{R,S,T}$.  Note that we cannot have two elements in
a triad such that $\var(S_1) \subset \var(S_2)$ because removing $\var(S_2)$ would isolate $S_1$.
Thus $\set{R,S,T}$ is the unique triad of $\rats$.
On the other hand, $R$ is fully dominated in $\brats$, so we transform it to triad-free $\brx$
(\specificref{Eq.}{eqn:brx}).

We now show that $\rsp(q)$ is \NP-complete if $q$ has a triad.
Then we will show that otherwise $\rsp(q)\in \PTIME$ (\specificref{Cor.}{easy rsp*}). The proofs will take the same form as for
resilience, however the following proof is slightly more subtle than the
analogous result for resilience.

\begin{lemma}[Triads make $\rsp(q)$ hard]\label{triads are hard for rsp}
Let $q$ be an sj-free CQ where all fully dominated atoms are exogenous. If $q$ has a triad, then $\rsp(q)$
is \NP-complete. 
\end{lemma}

\begin{proof}
Depending on which of the following cases the query falls into,
we build a reduction to $\rsp(q)$ from $\rsp(q_\triangle),\rsp(\rats)$ or $\rsp(q_\Tri)$.
Let $\mathcal{T} = \set{S_0,S_1,S_2}$ be a triad in query $q$.

\emph{Case 1:}  There is no endogenous atom $A$ such that $\var(A) \subseteq\var(S_i) \cap \var(S_j)$, 
for some $i\ne j$.
We will show that $\rsp(q_\triangle) \leq \rsp(q)$.

Given $D,\vec t,k$ we must produce $D',\vec{t'},k'$ such that 
\myequation{case1}{
\begin{array}{rcl}
(D,\vec t,k)\in \rsp(q_\triangle) &\leftrightarrow& (D',\vec{t'},k') \in \rsp(q)\; .
\end{array}}

Note that we may assume that $\vec t = R(a_0,b_0)$ for some values $a_0,b_0$, i.e., that $\vec t$ is a
tuple from $R$, because we know that $\rsp(q_\triangle)$ is hard no matter which relation we choose
the tuple from (\specificref{Prop.}{lem:res_rsp}).

In this case, we construct $D'$ exactly as we did in
\specificref{Lemma}{hard part dichotomy} (Cases 1 or 2), and as we did there, we let $k'=k$.
The only difference is that we must define $\vec{t'}$ from $\vec t$.
This is easy: recall that $\vec t = R(a_0,b_0)$.  We let
$\vec{t'}=S_0(\angle{a_0b_0},a_0,b_0)$, i.e., the corresponding tuple of $S_0$.  Thus, we have
  exactly simulated $q_\triangle$ in $q$, so \specificref{Eq.}{case1} holds. 

\emph{Case 2:}  There is an endogenous atom $A$ and some $i\ne j$, such that $\var(A) \subseteq\var(S_i)
\cap \var(S_j)$, but only for a unique pair $i\ne j$.
We show that $\rsp(\rats) \leq \rsp(q)$.  Let the pair be $0,2$, i.e., 
 $\var(A) \subseteq\var(S_0) \cap \var(S_2)$. 

Again, we are given $D,\vec t,k$, where $\vec t = R(a_0,b_0)$.  We  produce $D',\vec{t'}$, but
now such that,

\myequation{case2}{
\begin{array}{rcl}
(D,\vec t,k)\in \rsp(\rats) &\Leftrightarrow& (D',\vec{t'},k) \in \rsp(q)\;.
\end{array}}

We produce $D'$ and $\vec{t'}$ exactly as in Case 1, and we again have that all the witnesses and
minimum contingency sets for $\rats$ wrt $D,\vec t$ are preserved for $q$ wrt $D',\vec{t'}$.
Thus \specificref{Eq.}{case2} holds.

Finally, we are left with,

\emph{Case 3:}  There are endogenous atoms $A,B$ such that WLOG 
$\var(A) \subseteq\var(S_0) \cap \var(S_2)$, and 
$\var(B) \subseteq\var(S_0) \cap \var(S_1)$.

We know that $S_0$ is not fully dominated.  Thus, there must exist a non-solitary variable $w\in
\var(S_0)$ such that $w\not\in \var(A) \cup \var(B)$.  
Since $w$ is not fully dominated, there must be an endogenous atom  $C\ne S_0$ such that 
$C$ is reachable from $S_0$ without using edges from $\var(A) \cup \var(B)$.  Thus we have located a
tripod sitting in the hypergraph of $q$ 
(see \autoref{case3 fig}).  
It thus follow from \specificref{Prop.}{prop:tripodQuery}, that $\rsp(q)$ is \NP-complete as well.
\end{proof}

\begin{figure}
\begin{center}
\begin{tikzpicture}[ scale=.4]
\node[circle,draw=black, inner sep=0pt, minimum size=3cc, line width=1pt] (S0) at(0,0)  { $S_0$};
\node[circle,draw=black, inner sep=0pt, minimum size=3cc, line width=1pt] (S1) at(8,-8)  { $S_1$};
\node[circle,draw=black, inner sep=0pt, minimum size=3cc, line width=1pt] (S2) at(-8,-8)  { $S_2$};
\node[circle,draw=black, inner sep=0pt, minimum size=2.1cc, line width=1pt] (A) at(-4,-4)
     {\color{red} $A$};
\node[circle,draw=black, inner sep=0pt, minimum size=2.1cc, line width=1pt] (B) at(4,-4)
     {\color{dg} $B$};
\node[circle,draw=black, inner sep=0pt, minimum size=2.1cc, line width=1pt] (C) at(0,5)
     {\color{blue} $C$};
\foreach \from/\to in {S0/A,S0/B,A/S2,S2/S1,S1/B}
\draw[line width=2pt,color=black] (\from) -- (\to);
\foreach \from/\to in {S0/C}
\draw[dotted,line width=2pt,color=black] (\from) -- (\to);
\draw (-3.75,2.5) node { no edges from};
\draw (3.75,2.5) node { $\var(A)\cup\var(B)$};	
\end{tikzpicture}
\end{center}

\caption{Case 3 of the proof of \autoref{triads are hard for rsp}.  There is a tripod sitting in the
  hypergraph of $q$.}
\label{case3 fig}
\end{figure}
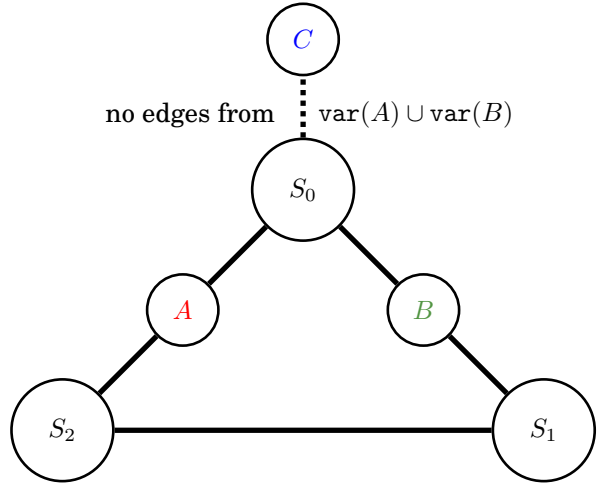

\subsection{The polynomial case}

As we saw in the previous section, the presence of triads in a query makes its responsibility
problem \NP-complete.  In the responsibility setting we require
full domination to make an atom exogenous.  This means that more atoms may remain endogenous, so
there can be more triads.  The query $\rats$ is an example: for resilience we use domination and
after applying domination, $\rats$ has no triads and thus $\res(\rats)\in\PTIME$.
However, if we may only apply full domination, then $\rats$ keeps the triad $R,S,T$ and thus
$\rsp(\rats)$ is \NP-complete.  

We now want to prove the polynomial case for responsibility.
Recall that in the proof of 
\specificref{Lemma}{easy part dichotomy}, we showed the following:

\begin{corollary}\label{ptime cor}
Let $q$ be a CQ that has no triad.  Then we can transform $q$, via a series of
dissociations, to a linear query $q'$ .
\end{corollary}

Then, since dissociations cannot make the resilience problem of an sj-free CQ easier 
(\specificref{Lemma}{fact: dissociation do not decrease complexity}), it followed that $\res(q)\in\PTIME$ for any such
triad-free query, $q$.

To prove that for any triad-free, sj-free CQ, $q$, $\rsp(q)\in\PTIME$, it suffices to prove that dissociations
cannot make the responsibility problem of such queries easier.
As we see next, there is a surprising complication to this proof,
which gives us an unexpected bonus result.

\subsection{A generalization of responsibility}\label{rsp star}

We want to prove that if $q'$ is obtained from $q$ through dissociation, then 
$\rsp(q) \leq \rsp(q')$.  
In the proof of the similar result for resilience we did the following.  We let $R^\exSymb(\vec
z)$ be the atom that was changed to ${R^\exSymb}'(\vec z,v)$.  We then reduced $\res(q)$ to $\res(q')$ by
mapping $(D,k)$ to $(D',k)$ where $D'$ is the same as $D$ with the exception that we let ${R}' =
\bigset{\vec t,d}{R(\vec t)\in D; d\in \dom(D)}$.  This transformation does not change the witness
set nor the contingency sets, because, by the way we formed ${R}'$ from $R$, the conjunct $R'(\vec
z,v)$ places the same restriction on $D'$ that $R(\vec z)$ places on $D$.

This proof goes through fine for responsibility except in one case, namely if the tuple $\vec t$
that we are computing the responsibility of belongs to $R$, the exogenous relation to which we have
added the new variable, $v$\footnote{The reader may wonder why we might need to compute the
  responsibility of an exogenous tuple.  The answer is that the tuple originally might have come from
  an endogenous relation which we transformed to an exogenous one using full domination.}.

When $\vec t \in R$, we would like to transform it to $\vec {t'}\in R'$ by appending a value, $a_i$,
corresponding to the new variable, $v$.  However, this will change responsibility in an unclear
way. In particular, the responsibility of $\vec t$ does not correspond to the responsibility of $\vec
t,a$ for any particular $a$.  It rather corresponds to the responsibility of $\vec t,a$ for \emph{all
  possible} $a$'s.  

To solve our problem, we need to generalize the notion of responsibility to include wildcards.

\begin{definition}[tuples with wildcards]\label{def wildcards}
Let $D$ be a database containing a relation, $R(x_1, \ldots, x_c)$.
Let $\tau =(s_1, \ldots, s_c)$ be a tuple such that each $s_i\in \dom(D) \cup \set{*}$, i.e.,
$\tau$ may have elements in the domain in some coordinates and the wildcard, $*$, in others.  We
call $\tau$ a \emph{tuple with wildcards}.  We say that a tuple $(a_1, \ldots, a_c)\in R$ \emph{matches}
$\tau$ iff for all $i$, $a_i = s_i$ or $s_i = *$.
When $D$ and $R$ are understood, $\tau$ represents a set
of tuples from $R$, 
$\angle{\tau} =  \bigset{\vec a \in R}{\vec a \mbox{ matches } \tau}$.
\end{definition}

For example, the tuple with wildcard, $(a,*)$, matches all pairs from $R$ whose first coordinate
is $a$.  We generalize responsibility to allow us to compute the responsibility of a set of
tuples denoted by a tuple with wildcards:

\begin{definition}[$\rsp^*$]\label{def: responsibility generalized}
Let $D$ be a database containing a relation, $R$, $q$ a query for $D$ and $\tau$ a tuple with
wildcards. Then $(D,\tau,k) \in \rsp^*(q)$ iff there exists a contingency set $\Gamma$ of size $k$
such that $(D-\Gamma) \models q$ and $(D -(\Gamma \cup \angle{\tau})) \not\models q$.
\end{definition}

Since $\rsp^*(q)$ is just a generalization of $\rsp(q)$ it is immediate that  $\rsp(q) \leq
\rsp^*(q)$.  Thus, $\rsp^*(q)$ is \NP-complete whenever $\rsp(q)$ is:

\begin{corollary}\label{hard rsp*}
Let $q$ be an sj-free CQ all of whose fully dominated atoms are exogenous.  It $q$ has a triad then
$\rsp^*(q)$ is \NP-complete. 
\end{corollary}

From our previous discussion, it now follows that dissociation does not make $\rsp^*(q)$ easier:

\begin{lemma}\label{lem: dissociation rspStar}
If $q'$ is obtained from $q$ through dissociation, then $\rsp^*(q) \leq \rsp^*(q')$.
\end{lemma}

Furthermore, linear queries are still easy for responsibility:

\begin{lemma}\label{linear easy rspStar}
 For any linear sj-free CQ $q$, $\rsp^*(q)$ is in \ptime.
\end{lemma}

\begin{proof}
The proof is a small modification of the proof for \specificref{Fact}{linear query is easy}.  
As before, we use network flow to compute the min cut over all $w$ extending any 
element of $\angle{\tau}$ of the network, $N_{\vec w,\tau}(q,D)$.
This new network has weight $\infty$ for every edge in $\vec w - \angle{\tau}$ and 0 for every edge
in $\angle{\tau}$.  See \autoref{fig:wildFlow}.
\end{proof}

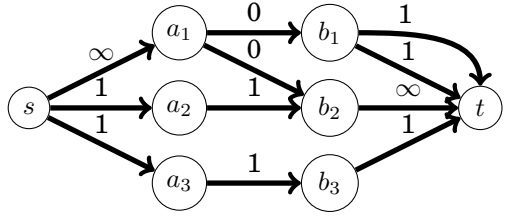
\begin{figure}
\begin{center}
\begin{tikzpicture}[ scale=.25,
	every circle node/.style={fill=white, minimum size = 4mm, draw},
	]
\node (s)  at 	(0,0) [circle] {$s$};
\node (a1) at 	(8,4) [circle] {$a_1$};
\node (a2) at 	(8,0) [circle] {$a_2$};
\node (a3) at 	(8,-4) [circle] {$a_3$};
\node (b1) at 	(16,4) [circle] {$b_1$};
\node (b2) at 	(16,0) [circle] {$b_2$};
\node (b3) at 	(16,-4) [circle] {$b_3$};
\node (t)  at 	(24,0) [circle] {$t$};
\foreach \from/\to/\l in {s/a1/$\infty$,s/a2/1,s/a3/1,a1/b1/0,a1/b2/0,a2/b2/1,a3/b3/1,b1/t/1,b2/t/$\infty$,b3/t/1}
\draw[->,line width=2pt,color=black] (\from) -- (\to) node[pos=.5,above]{\l};
\draw[->,line width=2pt,color=black] (b1) to [out=0,in=90] (t);
\node at (20,5) {1};
\end{tikzpicture}
\end{center}
\caption{$N_{\vec w,\tau}(q,D)$; $\;\vec w = A(a_1),R(a_1,b_2),S(b_2,c_2)$; $\; \tau = R(a_1,*)$.
This is an example of Network Flow in the proof of \specificref{Lemma}{linear easy rspStar} for query  $q
  \datarule A(x),R(x,y),S(y,z)$ and database $D=(A,R,S)$, where $A=\set{a_1,a_2,a_3}$,
$R=\set{(a_1,b_1),(a_1,b_2),(a_2,b_2),(a_3,b_3)}$,
  $S=\set{(b_1,c_1),(b_1,c_2),(b_2,c_2),(b_3,c_3)}$.  
\label{fig:wildFlow}}
\end{figure}

\begin{corollary}\label{easy rsp*}
If $q$ has no triad, then $\rsp^*(q)$ can be made linear by using dissociations, and is thus in \ptime. Therefore so is $\rsp(q)$.
\end{corollary}

We have thus proved our desired dichotomy for responsibility, and as a bonus, we have proved it for
responsibility with wildcards as well:

\begin{theorem}[Responsibility Dichotomy]\label{main rsp dichotomy thm}
Let $q$ be an sj-free CQ, and let $q'$ be the result of making all fully dominated atoms exogenous.
If $\mathcal{H}(q')$ contains a triad then $\rsp(q)$ and $\rsp^*(q)$ are \NP-complete.  Otherwise,
$\rsp(q)$ and $\rsp^*(q)$ are $\PTIME$.
\end{theorem}

It follows from \specificref{Lemma}{easy rsp*} and \specificref{Cor.}{hard rsp*} that $\rsp^*(q) \equiv \rsp(q)$ for all
sj-free CQ, $q$.  Note that it is not at all clear how one would build a reduction from $\rsp^*(q)$
to $\rsp(q)$.  However, our characterization of the complexity of $\rsp(q)$ and $\rsp^*(q)$ gives us
this result: 
After all fully dominated atoms are made exogenous, if there is a triad, then
$\rsp(q)$ is \NP-complete, thus so is $\rsp^*(q)$.  If there is no triad, then $\rsp^*(q) \in\PTIME$,
thus so is $\rsp(q)$:

\begin{corollary}\label{rsp* result}
For all sj-free CQ $q$, we have $\rsp(q) \equiv \rsp^*(q)$.
\end{corollary}

\subsection{Dichotomy for responsibility with FDs}\label{rspFD sec}

Our final theorem is that the dichotomy for responsibility continues to hold in the presence of FDs:

\begin{theorem}[FD Responsibility Dichotomy]\label{last thm}
Let $(q;\Phi)$ be an sf-free CQ with functional dependencies.  Let $(q^*,\Phi)$ be its
closure under induced rewrites, and such that all fully dominated atoms of $q^*$ are exogenous. 
If $q^*$ has a triad then $\rsp(q;\Phi)$ is \NP-complete.  Otherwise,
$\rsp(q;\Phi)\in\PTIME$.
\end{theorem}

\begin{proof}
Since FDs only make $\rsp(q)$ easier, we know that if $q^*$ has no triad then $\rsp(q^*)$ is easy,
thus so is $\rsp(q^*;\Phi)$ and thus also $\rsp(q;\Phi)$. For the converse, we show
that the reduction, $f$, from 
one of $\rsp(q_\triangle),\rsp(\rats),\rsp(q_\Tri)$ to $\rsp(q)$ which we built in 
\specificref{Lemma}{triads are hard for rsp} always produces databases, $D'$, that satisfy $\Phi$.  
The proof is almost exactly as in \specificref{Lemma}{hard part dichotomy with fd}.  Note that in the proof of
\specificref{Lemma}{triads are hard for rsp},
we use the same
reduction in all three cases, i.e., no matter if we are reducing from
$\rsp(q_\triangle),\rsp(\rats)$, or $\rsp(q_\Tri)$.
\end{proof}

\subsection{Using resilience to compute responsibility more efficiently} \label{sec: res applied to rsp}

We now show that in applications where we wish to find those tuples of highest responsibility, we can
find them more efficiently by computing resilience instead of responsibility.

Responsibility provides a measure of the causal contribution of an input tuple
to a query output. In prior work~\cite{MeliouGNS2011,MeliouGMS11}, in order to identify likely causes, we ranked
input tuples based on their responsibilities: tuples
at the top of the ranking are the most likely causes, whereas tuples low in
the ranking are less likely. Producing this ranking entails computing the
responsibility of every tuple in the database that is a cause for the query. This is computationally expensive,
and, ultimately, unnecessary: Since most applications only care about the
top-ranked causes, we only need to find the set $S_\rho$ 
consisting of the tuples of highest responsibility. 
Computing the responsibility of other tuples is unnecessary.
Using this insight, we can employ resilience to compute 
$S_\rho$ more efficiently than by calculating
the responsibility of every tuple in the database.

Even though resilience is strictly easier to compute than responsibility, we 
can compute  $S_\rho$, the set of tuples of highest responsibility,
by repeatedly computing resilience. The first
observation is that any minimum contingency set for resilience is contained in $S_\rho$.

\begin{proposition}\label{prop: rsp from res}
  As above, let  $S_\rho$ be the set of tuples of highest responsibility for a database $D$
  satisfying a binary query $q$.  Let $\Gamma$ be a minimum contingency set for $(q,D)$.
Then all members of $\Gamma$ have maximum responsibility for $D\models q$, i.e., $\Gamma\subseteq S_\rho$.
\end{proposition}

\begin{proof}
Let $q,D,S_\rho,\Gamma$ be as in the statement of the proposition.  Let $k= \abs{\Gamma}$.
Let $\vec t$ be any element of $\Gamma$.  Note that  $\Gamma - \set{\vec t}$ is a contingency set of
size $k-1$ for the responsibility of $(q,D,\vec t)$.  Suppose for the sake of contradiction that
some tuple $\vec {t'}$ had strictly greater responsibility than $\vec t$.  Then there must be a
contingency set $\Gamma'$ for the responsibility of $(q,D,\vec {t'})$ such that $\abs{\Gamma'} <
k-1$.  However, this means that $\Gamma'\cup \set{\vec {t'}}$ is a contingency set for the resilience
of $(q,D)$ of size less than $k$, contradicting the fact that $\Gamma$ is a minimum
contingency set.
\end{proof}

Therefore, all tuples in a minimum contingency set for resilience have maximum responsibility.
However, there may be additional tuples with maximum
responsibility that are not part of the selected resilience set $\Gamma$.
These can also be derived by a simple algorithm based on the 
following observation.

\begin{observation}
  Let $q,D,S_\rho,\Gamma,k$ be as in the proof of Prop. \ref{prop: rsp from res} and let $\vec{t'}$
  be any tuple in $D$.  Let $\Gamma'$ be a minimum contingency set for the resilience of
  $(q,D-\set{\vec{t'}})$.
Then $\vec{t'}\in S_\rho$ iff $\abs{\Gamma'}=k-1$.  Furthermore, if $\abs{\Gamma'}=k-1$ then 
$\Gamma' \subseteq S_\rho$.
\end{observation}

Thus, even though responsibility is 
harder to compute than resilience (\specificref{Lemma}{lem:res_rsp}), the following algorithm
computes the set of tuples of maximum responsibility by repeatedly computing resilience.

\begin{myalgorithm}[Computing max responsibility set, $S_\rho$, using resilience]\label{alg: rsp from res}
\begin{enumerate}\itemsep=0pt
\item Let $C$ be the set of causes of $D\models q$
\item Let $\Gamma$ be a minimum contingency set for $(q,D)$
\item $k := \abs{\Gamma}; \; S := \Gamma$
\item \textbf{for each} $\vec c \in C - S$:
\item \hspace*{.25in} Let $\Gamma'$ be a minimum contingency set for $(q,D-\set{\vec c})$
\item \hspace*{.25in} \textbf{if} $\;\abs{\Gamma'}=k-1$: \quad $S := S \cup \Gamma' \cup \set{\vec c}$
\item \textbf{return}$(S)$
\end{enumerate}
\end{myalgorithm}

%% file: full_6_FelatedWork.tex
\section{Related Work}\label{sec:relatedWork}
\specificref{Sections}{sec:intro} and~\ref{sec:background} have extensively discussed prior work and the connections between resilience, 
deletion propagation and responsibility~\cite{Buneman:2002,Cong12,Kimelfeld12,KimelfeldVW12}.
In this section, we discuss additional related work.

\textbf{Data provenance.}
Data provenance studies formalisms that can characterize the 
relation between the input and the output of a given query~\cite{DBLP:conf/icdt/BunemanKT01,DBLP:journals/ftdb/CheneyCT09,DBLP:journals/tods/CuiWW00,GKT07-semirings}.
Among the kinds of provenance, ``Why-provenance'' is the most closely related to resilience in databases.
The motivation behind Why-provenance is to find the ``witnesses'' for the query answer, i.e., the tuples or group of tuples in the input 
that can produce the answer. Resilience, searches to find a \emph{minimum} set of input tuples that can make a query false.

\textbf{View updates.}
The view update problem is a classical problem studied in the database literature~\cite{Bancilhon81,Cong12,Cosmadakis84,Dayal82,Fagin83,Keller85}.
In its general form, the problem consists of finding the set of operations that should be applied to the database in order to
obtain a certain modification in the view. Resilience and deletion propagation are a special cases of view updates.

\textbf{Causality.}
The study of causality is important in many areas other than databases, for
example in Artificial Intelligence and philosophy. Although an intuitive
concept, it is difficult to formally define causality and many authors have
presented possible definitions of causality. In our prior work, the notions of
causality and responsibility were strongly inspired by the work of Halpern and
Pearl~\cite{ChocklerH04,HalpernPearl:Cause2005}. Causal reasoning is based on
the idea of \emph{interventions}: understand how changes of input variables
affect an outcome, and thus relates in spirit to resilience. In the case of
resilience, the intervention is the deletion of input tuples. In
\Autoref{sec:outlook} we provide some additional discussion on how resilience
can address some applications of causality, and it has the benefit that it is
easier to compute than responsibility.

\textbf{Explanations in Databases.}
Providing explanations to query answers is important because it can help identify inconsistencies and errors
in the data, as well as understand the data and queries that operate on it.
Causality can provide a framework for explanations of 
query results~\cite{MeliouGMS11,MeliouGNS2011}, but it relies on the computation of responsibility, which is a harder problem than resilience.
Other work on explanations also applies interventions, but on the queries instead of the data~\cite{SudeepaSuciu14,DBLP:journals/pvldb/0002M13}. These approaches, try to understand how the deletion, addition, or modification of predicates may affect the result of a query.
There are also other approaches on deriving explanations that focus on specific database applications 
\cite{Agarwal+2007,Barman+2007,sigmod14-bender,Fabbri+2011,Khoussainova+2012,Thirumuruganathan+2012}. Finally, the problem of explaining \emph{missing} query results \cite{DBLP:conf/sigmod/ChapmanJ09,DBLP:journals/pvldb/HerschelH10,DBLP:journals/pvldb/HuangCDN08,DBLP:journals/pvldb/HerschelHT09,DBLP:conf/sigmod/TranC10} is a problem analogous to deletion propagation, but
in this case, we want to add, rather than remove tuples from the view. In this
paper, we focused the definition of resilience with respect to tuple
deletions; extending it to handle other kinds of updates is the topic of
future work.

%% file: full_7_Discussion.tex
\section{Discussion and outlook}\label{sec:outlook}
\introparagraph{Summary} This paper presents dichotomy results for the resilience and responsibility of sj-free
conjunctive queries. Our results extend and generalize previous complexity results
on the problem of \emph{deletion propagation with source side-effects} and
causal responsibility.

\introparagraph{Approximation for resilience of sj-free conjunctive queries} The dichotomy results we establish in this work define sets of queries for which we can solve resilience
in polynomial time, and sets of queries for which the problem is \NP-complete. We cannot hope to
find an efficient algorithm for the latter, unless $\mathsf{P}=\NP$, but we can look for an
approximation for the optimal solution. In particular, a constant factor approximation might be also
useful for finding a good approximation for the responsibility problem (see \specificref{Section}{sec: res applied to rsp}). 

\introparagraph{Conjunctive queries with self-joins} In order to complete the study of the complexity of resilience for conjunctive queries, we need to investigate the complexity 
of queries with self-joins. It is known that the problem is \NP-complete for a query as simple as $q
\datarule S(x),R(x,y),S(y)$ \cite{MeliouGMS11}. We suspect that the insights using triads to
characterize the complexity of resilience in the absence of self-joins may still be useful 
in the presence of self-joins.

\introparagraph{Unions of conjunctive queries} It would also be quite interesting to understand the
complexity of computing the resilience for queries that are unions of conjunctive queries, i.e., disjunctions of
conjunctions.  This is a natural extension which we started to explore when trying to generalize our
results about resilience to responsibility.  In particular, there is a natural way to view the
responsibility of a query as the resilience of a union of related queries.

%% file: full_Nomenclature.tex
\section{Nomenclature}\label{sec:nomenclature}

\begin{table}[h!]
\centering
\small
\begin{tabularx}{\linewidth}{  @{\hspace{0pt}} >{$}l<{$}  @{\hspace{2mm}} X @{}} 	
    \toprule
    \multicolumn{2}{l}{\textbf{Notation table}}\\
    \midrule
    D                   & database instance, union of all tuples in the relations, i.e., $D = \bigcup_i R_i$ \\ 
    A_1,\dots, A_m  	& atoms \\  
    A_i^\enSymb, A_i^\exSymb  & endogenous or exogenous atom \\
    \en{D}             & set of endogenous tuples: $\en{D} \subseteq D$ \\
    \ex{D}              & set of exogenous tuples: $\ex{D} = D \setminus \en{D}$    \\  
    D\models q          & $q$ is \texttt{true} in $D$\\
    D\not\models q      & $q$ is \texttt{false} in $D$\\
    \Gamma              & contingency set: subset of endogenous input tuples. $\Gamma \subseteq \en{D}$ \\
    \vec t				& tuple	\\
    \res(q)             & the resilience problem of query $q$\\
    \rsp(q)             & the problem of causal responsibility for query $q$\\
    \source(q)          & deletion propagation with source side-effects\\
    \view(q)            & deletion propagation with view side-effects\\
    q_\triangle         & triangle query $q_\triangle\datarule R(x,y),S(y,z),T(z,x)$\\
    q_\Tri              & tripod query $q_\Tri\datarule A(x),B(y),C(z), W(x,y,z)$\\
	\rats				& rats query $\rats \datarule A(x),R(x,y),S(y,z),T(z,x)$ \\
	\brats				& brats query $\brats \datarule A(x),R(x,y),B(y),S(y,z),T(z,x)$ \\	
    \varphi, \Phi       & a functional dependency (FD), or a set of FDs\\
    \mathcal{H}         & dual hypergraph (or simply hypergraph, in short)   \\
    q^*                 & closure of $q$ under induced rewrites\\
    	\var(A_i)	& set of all variables occurring in atom $A_i$ \\
	\var(q) 	& set of all variables occurring in query $q$\\
    \mathcal{T}         & triad  \\
    \dom(D)				& set of domain elements of $D$ \\
	\langle ab \rangle	& concatenated new domain values \\
	\tau					& tuple with wildcards \\
	\rsp^*(q)			& generalization of $\rsp(q)$ that computes responsibility of tuples with wildcards $\tau$  \\
    \bottomrule
    \end{tabularx}
\end{table}